\documentclass[10pt, journal]{IEEEtran}
\usepackage{lipsum}
\usepackage{epsfig,amsmath,url,bm}
\usepackage{cite}
\usepackage{epstopdf}
\usepackage{amsmath,amssymb,graphicx}
\usepackage{authblk}
\usepackage{balance}
\usepackage{amsthm}
\usepackage{caption}
\usepackage{xfrac}
\usepackage{subfig}
\usepackage{epsf,algorithm,algorithmic}
\newtheorem{ntheorem}{Theorem}[section]

\newtheorem{definition}[ntheorem]{Definition}

\title{Learning Based Methods for Traffic Matrix Estimation from Link Measurements}
\author[1]{Shenghe Xu}
\author[2]{Murali Kodialam}
\author[2]{T.V. Lakshman}
\author[1]{Shivendra Panwar}
\affil[1]{NYU Tandon School of Engineering}
\affil[2]{Nokia Bell Labs}

\begin{document}


\maketitle
\author{}

\begin{abstract}
Network traffic demand matrix is a critical input for capacity planning, anomaly detection and many other network management related tasks. The demand matrix is often computed from link load measurements. The traffic matrix (TM) estimation problem is the determination of the traffic demand matrix from link load measurements. The relationship between the link loads and the traffic matrix that generated the link load can be modeled as an under-determined linear system and has multiple feasible solutions. Therefore, prior knowledge of the traffic demand pattern has to be used in order to find a potentially feasible demand matrix. In this paper, we consider the TM estimation problem where we have information about the distribution of the demand sizes. This information can be obtained from the analysis of a few traffic matrices measured in the past or from operator experience. We develop an iterative projection based algorithm for the solution of this problem. If large number of past traffic matrices are accessible, we propose a Generative Adversarial Network (GAN) based approach for solving the problem. We compare the strengths of the two approaches and evaluate their performance for several networks using varying amounts of past data.
\end{abstract}



\section{Introduction}
The amount of traffic incident on a network is usually captured in the form 
of a traffic matrix (TM). A TM consists of the amount of traffic between 
each node pair in a network. Knowledge of the traffic matrix is essential to solving networking problems including link capacity planning, routing path design and network anomaly detection.  However, it is not easy for a
network operator to directly measure the point to point traffic in a network.
The most commonly used method to estimate the traffic matrix is to use link 
load measurements to infer the traffic matrix. The amount of traffic on a link is relatively easy to measure or estimate using traffic monitoring 
mechanisms like NetFlow.  

In a network with $n$ nodes, the size of the traffic matrix is $O(n^2)$ whereas the number of links in the network is typically $O(n)$. Therefore,
the problem of determining a traffic matrix from link load measurements is
deriving a solution to an under-determined system of linear
equations. This system has an infinite number of solutions even if we 
restrict the solutions to be non-negative. Therefore some additional information has to be used to restrict the solution space to this system and 
obtain a single traffic matrix. This additional information or extra knowledge typically takes the form of assuming some spatial or temporal correlations about the entries in the traffic matrices. We give two examples of these assumptions, one spatial and one temporal.
\begin{itemize}
    \item {\em Gravity Models} where a weight is associated with each node in the network and the amount of traffic between two nodes is proportional to the product of the weights. This reduces the dimension of the search space from $O(n^2)$ to $O(n)$ (the weight associated with each node).
    \item{\em Proportional Splitting} where it is assumed that the traffic
    from a given node is split proportionally to different destinations and these proportions are time invariant. Though there are still $O(n^2)$
    parameters, data can be collected across $n$ time periods and this data can be jointly used to solve for the proportions
\end{itemize}

Another class of assumptions is traffic sparsity in certain transform domain.
See \cite{vardi1996network,barford2002signal,lakhina2004structural,zhang2003fast,zhang2005estimating, zhang2009spatio, gursun2012traffic, chen2014robust,xie2015sequential,xie2016accurate, xie2017accurate} for examples of different assumptions for deriving a unique traffic matrix from link measurements. 

In this paper, we consider restrictions on the traffic matrix estimation problem that arises from traffic matrix observations. If the operator has measured 
a few traffic matrices on the network of interest or some similar network,
then it is reasonable to restrict the estimated traffic matrix to have properties similar to the measured or observed traffic matrices. 

\subsubsection*{\bf Distribution Constraint}
It has been observed in practice \cite{fukuda2008towards, downey2001evidence, chandrasekaran2009survey} that the point-to-point traffic in a network
is generally not uniform. There are some large traffic source-destination pairs, and several low traffic source-destination pairs. Modeling the demand size variation as a distribution, the objective of the traffic matrix estimation problem is to determine a traffic matrix that achieves the measured link load and follows the given distribution.
\subsubsection*{\bf Similarity Constraint}
More generally, if we are given a previously observed set of traffic matrices, the objective of the traffic matrix estimation problem is to derive a traffic matrix that achieves the given link load and is "similar" to the previously observed traffic matrices. In this case, it is possible to capture more complex spatial correlations between different source-destination pairs in the traffic matrix. 
The problem of determining a solution to a under-determined linear system has been studied in the signal processing literature \cite{donoho2006compressed}. One way of getting unique recovery is to assume sparsity and 
the objective is to determine a solution to the linear system with the minimum number of non-zero 
components or a solution that minimizes the $L_1$-norm.
More recently, there has been work to construct a solution to a linear system that is close to the range space of a generative model\cite{bora2017compressed}. The generative model can be specified by a
Generative Adversarial Network (GAN)\cite{goodfellow2014generative, gulrajani2017improved, arjovsky2017wasserstein} or a Variable Autoencoder \cite{kingma2013auto}. We make use of these new approaches to derive solutions to the 
TM estimation problem.
\subsection{Our Contributions}
In this paper we propose two methods to solve the TM estimation problem that takes into account the structure of the traffic matrix.
\begin{itemize}
    \item In the case where one or a few prior traffic matrices are available, we develop an iterative projection based method to find a solution to the system $\bm{Ax=b}$ where the solution $\bm{x}$ satisfies an empirical distribution that is derived from the prior traffic matrices. To our knowledge, this is the first work that determines the solution of an under-determined system where the solution has to satisfy a distribution constraint.
    \item For the case where there are many prior traffic matrices, we develop a GAN based approach that "learns" these characteristics of these traffic matrices and then derives a solution to the system that is "similar" to the previously observed traffic matrices. 
\end{itemize}
The rest of the paper is organized as follows. Section \ref{Sec:related} briefly summarizes related work. In Section \ref{Sec:problem} we formulate the problem. The projection based method is proposed in Section \ref{Sec:projection}. In Section \ref{Sec:GAN} and Section \ref{Sec:GAN_method} we introduce the GAN based TM estimation method. Experiment setup is included in Section \ref{Sec:experiment}. The performance of the methods is evaluated in Section \ref{Sec:performance}. In Section \ref{Sec:conclusion} we draw the conclusions and propose directions for future work. 
\section{Related Work}\label{Sec:related}
Traffic matrix estimation also called network tomography is an extremely important first step for several network design and network management problems. This problem has been studied extensively under different assumptions about traffic demand information and estimation.  
An example of research exploiting temporal correlation to estimate the TM is \cite{vardi1996network}, where it is assumed that 
the traffic demands over time follow Poisson distribution and this information is used to derive a traffic matrix. 
Several papers \cite{lakhina2004structural, zhang2003fast, zhang2005estimating, barford2002signal} consider using spatial characteristics of the TMs to improve the recovery results. Zhang et al. \cite{zhang2003fast} proposed gravity models to solve the problem of network tomography.  In \cite{zhang2005estimating}, the authors proposed an information-theoretic method for network tomography. 
Later works \cite{zhang2009spatio,xie2016accurate,xie2017accurate} consider using both spatial and temporal information for better recovery results.  A compressive sensing based method called Sparsity Regularized Singular Value Decomposition (SRSVD) was introduced in 
\cite{zhang2009spatio}. In addition to link measurements, measurements of demands between some of the origins and destinations are assumed to be available.  Measurements of previous time slots are also used to improve estimation accuracy. The SRSVD utilizes sparsity of TMs in transform domain for recovery. There is also additional literature 
\cite{gursun2012traffic, chen2014robust, xie2015sequential} that utilizes low rank or sparse characteristics of TMs for TM completion. 
Instead of forming sampled TMs into a 2D matrix, \cite{xie2016accurate, xie2017accurate} proposed to form TMs directly into 3D tensors. In this way the periodicity of certain traffic demand features can also be utilized by tensor completion methods for traffic demand estimation. 
More recently, deep neural networks (DNNs) \cite{lecun2015deep, hinton2006reducing} have achieved some of the state-of-the-art results in areas including image inpainting \cite{xie2012image} and image compressive sensing \cite{bora2017compressed}. Since TM estimation is also a similar problem, neural networks have also been used in this area. In \cite{nie2016traffic} the authors proposed to use DNN for traffic matrix completion. In \cite{zhao2019spatiotemporal} the authors proposed to use neural networks including DNN, convolutional neural networks \cite{lecun2015deep} and long short-term memory \cite{hochreiter1997long} with wavelet decomposition \cite{pati1993orthogonal} for TM prediction. 
All of these methods utilize certain spatial or temporal correlation in the traffic demands to obtain suitable estimates of TMs. The main contribution of the paper is the problem of TM estimation when the only information available is the distribution of the demand sizes. To our knowledge, this problem has not been addressed in the literature and therefore none of the techniques developed in the literature can be used for this problem. If data is sufficient, the GAN based method can also capture spatial correlations in the TMs for better reconstruction results. 

\section{Problem Definition}\label{Sec:problem}
Assume that the network is represented as a directed capacitated graph $G=(V,E)$ with 
$n$ nodes $V$ and $m$ directed links $E$.  Assume that we are given the set of link weights 
$\mathbf{w} = \left( w(e_1), w(e_2), \ldots, w(e_m) \right).$ The traffic in the network is specified in terms of a $n \times n$ traffic matrix between each pair of nodes in the network. The traffic between source node
$s$ and destination node $d$ is represented by $x_{sd}.$ In general, there may not be traffic between all source-destination pairs. We use $p$ to denote the number of source-destination pairs between which there is non-zero traffic. In the rest of the paper, instead of viewing the traffic as a matrix, we represent the traffic 
as a $p$-vector $\bm{x}$.
For a given set link weights $\mathbf{w}$, traffic is routed between nodes $s$ and $d$ along the shortest path between $s$ and $d$. We assume that ties between shortest paths are broken arbitrarily. It is easy to extend the approach in this note to the case where traffic is split between equal cost paths (ECMP). This routing induces a flow on the links in the network. Let $S(e)$ denote the set of source destination pairs that are routed on link $e$. A source-destination pair $(s,d) \in S(e)$ if link $e$ is on the shortest path from $s$ to $d$.  Let $b(e)$ denote the measured flow on link $e$. The traffic matrix estimation problem is the determination of
$x_{s,d}$ given the link load measurements $b(e)$.
Note that the traffic flow on link $e$ 
\begin{equation}
  b(e) = \sum_{(s,d) \in S(e)} x_{s,d}  
\end{equation}
We create a {\em routing matrix} $A$ with $m$ rows, one corresponding to each directed link, and $p$ columns, one corresponding to each source-destination pair. We set 
\begin{equation}
A_{ij}=
\begin{cases}
1 \quad \mbox{ if } M(j) \in S(i)\\
0 \quad \mbox{ Otherwise } 
\end{cases}
\end{equation}
where $M$ is the mapping from row index $i$ to a source-destination pair $(s, d)$.
The objective of the TM estimation problem is to determine a {\em non-negative} solution 
to the system $\bm{Ax=b}$ 
where $\bm{A}$ is an $m \times p$ routing matrix and $\bm{b}$ is the link load vector.
If there is no additional information, the number of source-destination will be much more than the number of links,  then this system has an infinite number of solutions since $m \ll p$. Therefore, we impose additional constraints on $\bm{x}$ in order to narrow down the solution space.
\subsection{Distribution Constraint}
In order to motivate the distribution constraint, we consider the traffic matrix estimation problem on a network (NET82) with $82$ nodes and $296$ directed links.  Each demand matrix comprises of $6724=(82 \times 82)$
potential demands. The NET82 dataset is a real network with available measurements of the real TM. In the demand matrix that was measured, there are $1939$ non-zero demands. We show a plot of demand sizes on the left side of Figure \ref{demdes}. 
 \begin{figure}[htb]
\centering
\includegraphics[width=3.8 in]{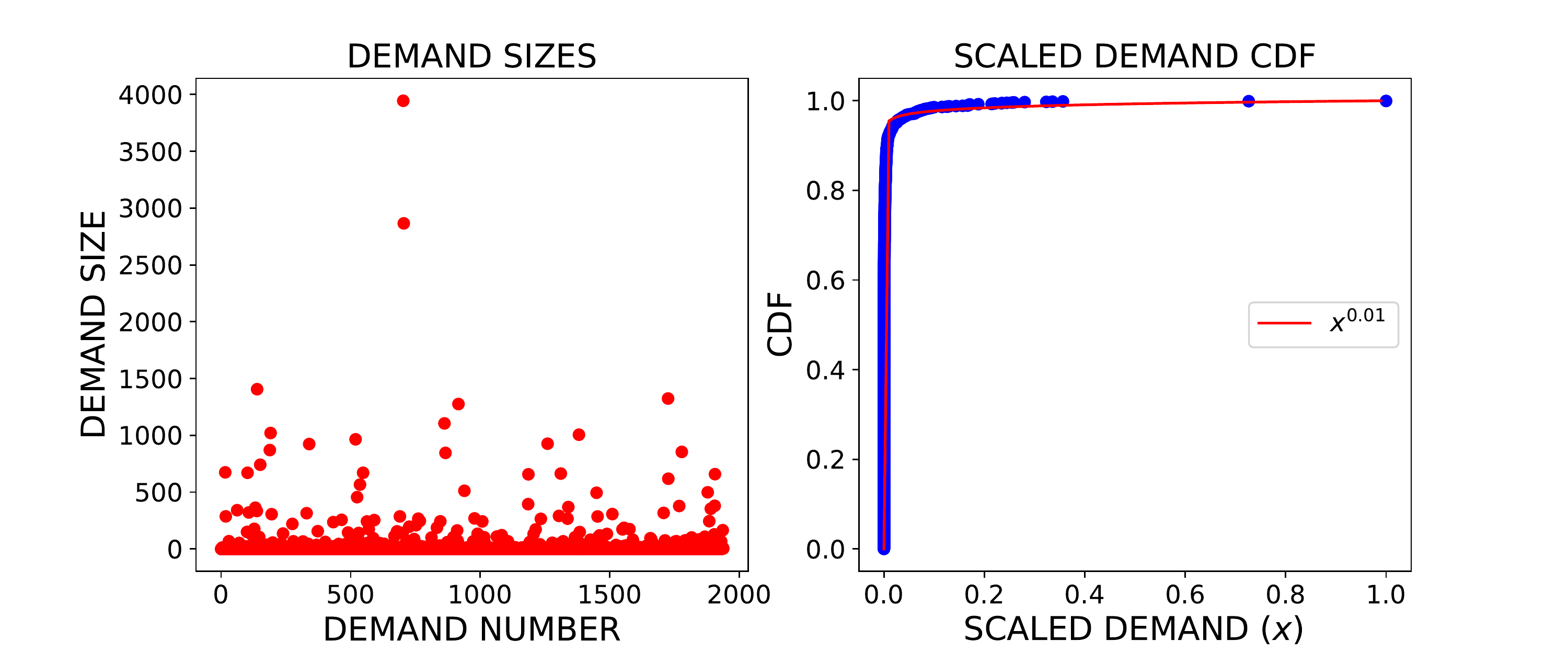}
\caption{Plot of the Demands and the normalized Empirical Distribution Function}
\label{demdes}
\end{figure}
Note that there are a few large demands and several medium to small demands.
The right hand side of Figure \ref{demdes} shows the cumulative distribution 
function of the normalized demand sizes where the demands are scaled such that the largest demand is one unit. Note that that cdf is modeled well using a power law distribution $x^{0.01}.$
The same pattern is observed in $4$ other demand matrices on the same network.
Therefore, when estimating a TM on this network from link load measurements, we would ideally like this demand matrix to have the same pattern of demands. Assume 
that we have observed a link load vector $\bm{b}$ from an unknown traffic matrix 
and we find a solution for the system $\bm{Ax=b}, \bm{x} \geq 0.$ We show two alternative solutions to this system in Figure \ref{twosol}. 
 \begin{figure}[htb]
\centering
\includegraphics[width=3.8 in]{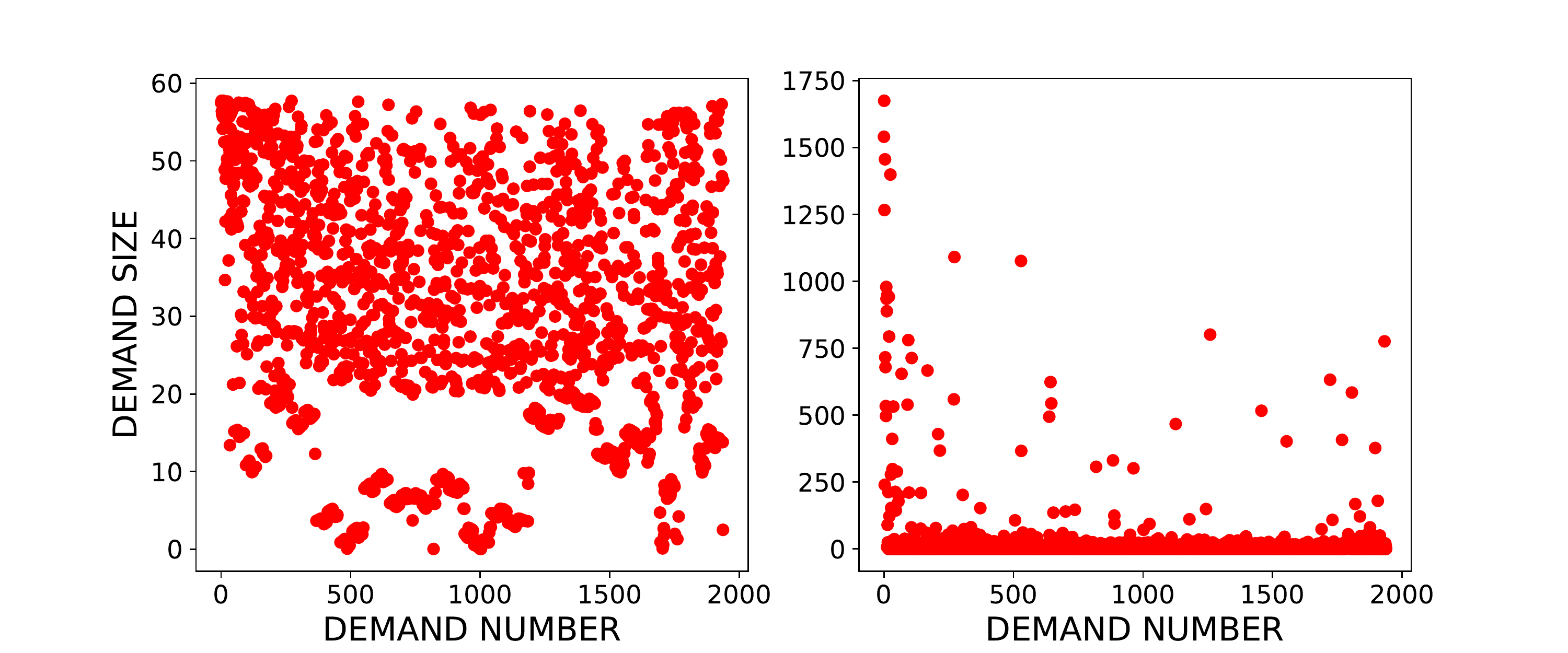}
\caption{Two Different Traffic Matrix Estimates for the Same Link Load Observation}
\label{twosol}
\end{figure}
In the solution on the left, the traffic matrix comprises of uniformly distributed demands and in the solution shown in the right hand side of the Figure \ref{twosol},the normalized demands follow the power law $x^{0.01}$. 
It is much more likely, given information about the demand distribution that the actual data 
looks like the traffic distribution on the right. We want to caution the reader that even with this additional restriction on the demand size distribution, the TM reconstruction may not be unique. 
In order to formally define the distribution constraint, we first define the 
the {\em empirical cumulative distribution function} for a given data set.
\begin{definition}
Given a set of data points $y_1 \leq y_2\leq  \ldots\leq y_n$, the empirical cumulative 
distribution function (empirical cdf) of these points is a step function that jumps up by $\frac{1}{n}$ at each of the $n$ data points. Its value at any specified value $z$, is the fraction of observations of the measured variable that are less than or equal to $z$.
\end{definition}
The empirical distribution function is an estimate of the cumulative distribution function that the points in the sample are generated from and it converges with probability one to the underlying cdf. 
\subsection*{Specifying the CDF of the Solution}
Once we observe one or a few traffic matrices, we can construct the empirical cdf of the raw demands. Since the total traffic in the network can change significantly over time, we have to normalize the demands and use the normalized cdf as shown in the right hand side of Figure \ref{demdes}.
We now define the normalized empirical cumulative distribution function of 
a set of observations.
\begin{definition}
Given a set of data points $y_1 \leq y_2\leq  \ldots\leq y_n$, the normalized empirical cumulative  distribution function (normalized empirical cdf) of these points is a step function that jumps up by $\frac{1}{n}$ at each of the $n$ scaled data points $\frac{y_1}{y_n}, \frac{y_2}{y_n},  \ldots,1$. Its value at any specified value $z \leq 1$, is the fraction of observations of the measured variable that are less than or equal to $z$.
\end{definition}
The domain of the normalized empirical cdf of a set of data points is $[0, 1].$
Assume that the observed traffic matrix has a normalized empirical cdf of $G(z)$
for $0 \leq z \leq 1$. As part of the solution procedure, we have to generate
random variables having a normalized empirical cdf of $G(z)$. A random variable having cdf $F(z)$ can be generated easily using standard random variable generation procedure. We want to use this process to generated random variables having a normalized empirical cdf of $G(z).$  The following result relates the cdf of a random variable to the normalized cdf of $n$ iid samples 
of the random variable.
\begin{ntheorem}
\label{rvth}
Let $X_1,X_2, \ldots, X_n$ represent independent, identically distributed samples from a probability density function $f(x)$ (with the corresponding distribution function $F(x)$). Let 
\begin{equation}
 Y_i = \frac{X_i}{\max_j X_j}   
\end{equation}
Then, $Y_i$ are distributed with cdf
\begin{equation}
    G(y) = n \int_t F(yt) \left[ F(t) \right]^{n-2} f(t) dt, \quad 0 \leq y \leq 1.
\end{equation}
\end{ntheorem}

\begin{proof}
Given $x_1,X_2, \ldots, X_n$ iid from a distribution function $F(x)$, we let
\begin{equation}
  M = \max_{1 \leq j \leq n} X_j.  
\end{equation}
Then 
\begin{equation}
 Pr \left[ M \leq t \right] = Pr \left[ X_j \leq t \forall j \right] = \left[ F(t) \right]^n,   
\end{equation}
with the corresponding density function $n \left[ F(t) \right]^{n-1} f(t).$
We set 
\begin{equation}
 Y_j = \frac{X_j}{n} \quad 1 \leq j \leq n.   
\end{equation}
Then
\begin{eqnarray}
Pr \left[ Y_j \leq y \right] & = & \int_t Pr \left[ X_j \leq y t | M = t \right] Pr \left[ M = t \right] \nonumber \\
& = &  \int_t Pr \left[ X_j \leq y t | X_j \leq t  \right] Pr \left[ M = t \right] \nonumber \\
& = & \int_t \frac{Pr \left[ X_j \leq y t \right]} {Pr \left[ X_j \leq t  \right]} Pr \left[ M = t \right] \nonumber \\
& = & n \int_t \frac{F[yt]}{F[y]}  \left[ F(t) \right]^{n-1} f(t) \nonumber \\
& = & n \int_t F(yt) \left[ F(t) \right]^{n-2} f(t) dt,  \\
&& \quad 0 \leq y \leq 1. \nonumber
\end{eqnarray}
\end{proof}
We now give an example of the use of this theorem that is also very useful in practice to generate samples with the desired normalized empirical cdf. In many examples, the normalized cdf of the demand sizes follows a power law with parameter $\alpha$. In this case, the $G(x) \thicksim x^{\alpha}$ for some specified value of $\alpha$ for $0 \leq x \leq 1.$ Note that the higher the value of $\alpha$ the smaller is the number if larger demands. In the next result, we use Theorem \ref{rvth} to show that the a suitable underlying
beta distribution has a normalized power law cdf. The probability density function of a beta distribution is given by
\begin{equation}
    f(x) = C x^{\alpha-1} (1-x)^{\beta - 1}
\end{equation}
where $C$ is a constant to ensure that the total probability is 1. 
This distribution covers a 
common case. It is possible to use the result of Theorem \ref{rvth} to generate any desired normalized empirical cdf. 
\subsection*{Normalized empirical cdf of a Beta Distribution}
If $X_i \thicksim B(\alpha,1)$ for $1 \leq i \leq n$ denote $n$ iid samples from a beta distribution with parameters $(\alpha,1)$ then the distribution and density functions of 
$X_i$ are 
\begin{equation}
 \quad \quad F(x)= x^\alpha, f(x)= \alpha x^{\alpha-1},   \quad 0 \leq x \leq 1.   
\end{equation}
Therefore from Theorem \ref{rvth}, the normalized cdf  is
\begin{eqnarray}
G(y) & = & n \int_0^1 F(yt) \left[ F(t) \right]^{n-2} f(t) dt \quad 0 \leq y \leq 1 \nonumber \\
& = & n \int_0^1 (yt)^{\alpha} t^{(n-2) \alpha} \alpha t^{\alpha-1} dt \quad 0 \leq y \leq 1  \nonumber \\
& = & n \alpha y^\alpha \int_0^1 t^{n \alpha -1} dt \quad 0 \leq y \leq 1  \nonumber \\
& = & y^\alpha \quad 0 , \leq y \leq 1
\end{eqnarray}
Note that the normalized empirical cdf is independent of $n$ and is only a function of $\alpha$. This is not true in general. Therefore, if we need 
to generate $n$ random variates having a normalized empirical cumulative cdf of $x^\alpha$ we do the following:
\begin{itemize}
    \item Generate $X_1,X_2, \ldots, X_n$ independent random samples from $B(\alpha,1).$
    \item Let $X_{\mbox{max}}= \max_{1 \leq i \leq n} X_i.$
    \item Output $$\frac{X_1}{X_{\mbox{max}}}, \frac{X_2}{X_{\mbox{max}}}, \ldots, \frac{X_n}{X_{\mbox{max}}}$$ as the set of $n$ samples with normalized empirical cdf $x^\alpha.$
\end{itemize}

\section{Proj-D:  Projection Based Traffic Matrix Estimation Method}\label{Sec:projection}
Kakcmarz method \cite{karczmarz1937angenaherte} or the Algebraic Reconstruction Technique (ART) is a 
well known technique for finding a feasible solution to the system 
$Ax=b.$ Assume that there are $m$ rows in the matrix and $p$ columns. Recall 
that each of the $m$ rows correspond to a link load measurement and 
each of the $p$ columns corresponds to a demand. 
We can represent the set of equations as $a_i x = b_i$ for $i,2, \ldots, m$
and $a_i$ and $x$ is a $p$ dimensional vectors.
ART is a cyclic projection technique where we start off from an arbitrary initial $p$-vector $x$. The algorithm then projects this point onto the first constraint $a_1 x = b_1.$ Projection just involves finding the closest point 
to $x$ on the hyperplane $a_1 x = b_1.$ This is the new point. 
This point is then projected onto the second hyperplane and so on until 
we reach hyperplane $m$. This point is then projected onto the first hyperplane and this process is repeated in a cyclic manner as shown in the 
Cyclic Projection Algorithm 
\begin{algorithm}
\label{cpa0}
\caption{Cyclic Projection Method}
\label{alg:Pro}
\begin{algorithmic}[1]
\STATE Pick an arbitrary $p$-vector  $\bm{x}.$ \\
\FOR{$k=1,2, \ldots, K$ }
\FOR{$i=1,2, \ldots, m$}
\STATE $\bm{x} \leftarrow \bm{x} + \bm{a_i}^T(b_i - \bm{a_i} \bm{x})/(\bm{a_i}\bm{a_i}^T)$
\ENDFOR
\ENDFOR
\end{algorithmic}
\end{algorithm}

\begin{ntheorem}
The Cyclic Projection Algorithm shown above converges to a feasible solution to $Ax=b$ after a sufficient number of iterations. 
\end{ntheorem}
See \cite{karczmarz1937angenaherte} for a proof of this result. 
In the description of the cyclic projection algorithm, we refer to one iteration through all $m$ constraints as a {\em cycle.}
This cyclic projection algorithm can be extended directly to the case where
we want to find a non-negative feasible solution to the system $Ax=b$
by modifying the projection step by 
\begin{equation}
   \bm{x} \leftarrow \max\left\{ 0, \bm{x} + \bm{a_i}^T(b_i - \bm{a_i} \bm{x})/(\bm{a_i}\bm{a_i}^T) \right\}  
\end{equation}
where the max operation is a pointwise maximum. In other words, if after 
computing the projection, some components of $\bm{x}$ are negative, then we set these components to zero.
More recently randomized versions of the cyclic projection method where the next hyperplane to project onto is picked at random has been shown to have 
linear expected convergence \cite{strohmer2009randomized}.
If we use the cyclic projection algorithm (or its randomized version), then
the method gets an arbitrary solution. In order to ensure that the 
solution satisfies the distribution constraint, we periodically move the current solution to the a compatible point in the distribution. This is done as follows:
\begin{itemize}
    \item Once every $t$ cycles, we take the current solution $x$ and assume 
    that we renumber the components such that $x_1 \leq x_2 \leq \ldots \leq x_p.$ 
    \item We generate a $p$ random variates $y_1 \leq y_2 \leq \ldots \leq y_p$ that have the desired normalized empirical distributon.
    For instance, if we want $x$ to have a power law distribution with power law exponent $\alpha$, then we generate n iid samples from a beta distribution $B(\alpha,1)$ and then $y$ is the ratio of the these iid samples to the maximum value in the iid samples.
    \item We set $x_i= \lambda y_i$ for $1 \leq i \leq p$ for a suitably
    chosen scaling parameter $\lambda$
\end{itemize}
The scaling parameter $\lambda$ is chosen to minimize the deviation $D$ where $D$ is defined as 
\begin{equation}
  D = \min \sum_{j=1}^m \left( \lambda a_j y - b_j \right)^2.  
\end{equation}
Note that $D$ is sum of the squared deviation over all the constraints. Using calculus, it is easy to see that the optimal solution is 
\begin{equation}
   \lambda = \frac{\sum_{j=1}^m ( a_j y ) b_j}{\sum_{j=1}^m ( a_j y ) ^2} 
\end{equation}
We now label all the $y$ values by $\lambda$ and map the $y$ variables to the corresponding $x$ variables, that is,  
$x_i = \lambda y_i$ for $1 \leq i \leq p.$
This is the new starting point for the next cycle. The algorithm is terminated after $K$ cycles. We can view this process as running the cyclic projection method with $K$ starting solutions having the desired normalized empirical cdf. The overall description of the algorithm is shown below.

\begin{algorithm}
\label{cpa1}
\caption{Proj-D: Projection Based TM Estimation}
\label{alg:Pro0}
\begin{algorithmic}[1]
\STATE Pick an arbitrary $p$-vector  $\bm{x}.$ \\
\FOR{$k=1,2, \ldots, K$ }
\FOR{$j=1,2, \ldots, t$ }
\FOR{$i=1,2, \ldots, m$}
\STATE $\bm{x} \leftarrow \max\left\{ 0, \bm{x} + \bm{a_i}^T(b_i - \bm{a_i} \bm{x})/(\bm{a_i}\bm{a_i}^T) \right\} $
\ENDFOR
\STATE Reorder $\bm{x}$ such that $x_1 \leq x_2 \leq \ldots \leq x_m$
\STATE Generate $y_1 \leq y_2 \leq \ldots \leq y_p$ with the 
desired normalized empirical distribution
\STATE Compute $ \lambda = \frac{\sum_{j=1}^m ( a_j y ) b_j}{\sum_{j=1}^m ( a_j y ) ^2}$
\STATE Set $x_i = \lambda y_i$ for $1 \leq i \leq p$
\ENDFOR
\ENDFOR
\end{algorithmic}
\end{algorithm}
When generating random variables having the desired normalized empirical distribution, we can repeat the generation of the random variables and finding the optimal $\lambda$ multiple times and pick the solution that has the minimum $D$ value.
If the value of $t$ is chosen to be large enough that the solution over two successive cycles over all the constraints does not vary the solution too much.
Proj-D assumes that the there is enough data or operator experience to specify the (normalized) distribution of the demands. If there are several prior traffic matrices available, then it is possible to not only capture the distribution information but also spatial correlations. We use a GAN based approach to address this problem. Though the GAN based approach does indeed capture spatial correlations in addition to any distribution information, the projection based approach that is tailor made for distribution problem out performs the GAN based approach if we only have distribution information.

\section{Generative Adversarial Networks}\label{Sec:GAN}
The idea of using a GAN based approach to capture spatial correlation in the traffic matrix was motivated by the  impressive capabilities demonstrated by GANs for generating samples that resemble real world images \cite{goodfellow2014generative, gulrajani2017improved, arjovsky2017wasserstein}. The training of a GAN involves a game between the generator network and discriminator network. The generator and discriminator are both neural networks. The generator learns a mapping from random noise to the space of the given signal. The discriminator tries to distinguish between the real signal and the generated signal. During the game of GAN training, the discriminator is updated by learning from the real and generated images. The generator is updated by the gradient provided by the discriminator so that the generator learns to generate samples that resemble the real images. 

The game between the generator $T$ and discriminator $D$ can be written as the objective:
\begin{equation}
    \min_{T} \max_{D} \mathbb{E}_{\bm{x}\sim \mathbb{P}_r} [log(D(\bm{x}))] + \mathbb{E}_{\widetilde{\bm{x}}\sim \mathbb{P}_t}  [1 - log(D(\widetilde{\bm{x}}))]
\end{equation}
where $\mathbb{P}_r$ is the distribution of real data and $\mathbb{P}_t$ is the distribution of the data generated from the generator network $T$.

The game involved in the training process of a GAN requires that there exists some kind of balance between the generator and discriminator. If the discriminator is too strong then it fails to provide useful gradient for the training of generator. Various kinds of methods have been proposed to stabilize the training process of GANs \cite{arjovsky2017wasserstein, gulrajani2017improved}. In \cite{arjovsky2017wasserstein}, the Wasserstein-1 distance was proposed for the training of GANs. In addition, the authors in \cite{gulrajani2017improved} proposed a gradient penalty approach for the training of GANs called WGAN-GP, which shows even better performance for the task of image generation. In this paper we adopt the method of WGAN-GP as the training process of the GAN. 

\section{Traffic Matrix Estimation Using Generative Adversarial Networks}\label{Sec:GAN_method}
Since GANs can capture the characteristics of given data, the authors in \cite{bora2017compressed} proposed to use a GAN as a mapping from latent space to signal space for the application of compressive sensing. Their results show that the GAN based compressive sensing method achieves better performance when the sampling rate of the signal is low. The problem of traffic demand matrix estimation given link measurement has the same format as the problem of image compressive sensing \cite{bora2017compressed}. 
Since the link measurements of a TM are also relatively low, we propose to solve the traffic matrix estimation problem with a GAN as the generator for the traffic matrix. Suppose the latent variable of the GAN is $\ell$, the generator $T$ generates the estimated traffic matrix $T(\ell)$, then the problem of traffic matrix estimation can be written as:
\begin{eqnarray}
\min_{\bm{l}} \| \bm{y} - \bm{A}T(\bm{\ell})\|^2_2.   
\end{eqnarray}
A properly trained generator $T$ provides a mapping from the lower dimensional latent space to the space of possible traffic matrices. Since the function $T$ is differentiable, the objective function can be optimized by simple gradient descent. 

Compared with the projection approach, the estimation method using a GAN can be applied for more general cases. If we only have knowledge of the normalized empirical distribution then the GAN can be trained with data generated from the given normalized empirical distribution. If measurements from the past are available, the GAN can also be trained with the data from the past. 

\subsection{Traffic Matrix Estimation Under A Distribution Constraint}
We first consider the problem of TM estimation under a distribution constraint. Unlike the assumption of signal sparsity from previous compressive sensing methods, which can be enforced by adding sparsity regularization terms to the objective function, it is unclear how a distribution constraint can be incorporated into the objective function. However, since a GAN is able to capture the characteristics of given data and generate samples with similar features, the distribution constraint can be included in the objective function by training a GAN that generates samples following a similar distribution. Then the optimization can be conducted in the latent space. Given the cost function
\begin{equation}
    L = \| \bm{y} - \bm{A}T(\bm{\bm{\ell}})\|^2_2,
\end{equation}
the gradient of $L$ can be easily computed by the chain rule. Therefore $L$ can be updated step by step by using simple stochastic gradient descent or any other optimizer such as the adaptive moment estimation (Adam) optimizer \cite{kingma2014adam}. For the experiments in this paper we use the Adam optimizer as the optimizer over the latent space. In the experiments, we find that choosing a better initial point in the latent space can help reduce the optimization steps and provide better estimation results. So we generate $N_i$ random vectors $\bm{n_i}$ in the latent space and select the one that provides link measurements that is closest to the given link measurements. The we run the optimization for $N_2$ steps. We show the details of this method in algorithm 3. This GAN based estimation method under a distribution constraint is denoted as GAN-D. 

\begin{algorithm}
\caption{GAN Based TM Estimation Method}
\label{alg:GAN}
\begin{algorithmic}[1]
\STATE Generate random Gaussian noise $\bm{n_0}$.\\
\STATE $\bm{\hat{n}} = \bm{\hat{n}_0}$
\FOR{$i=1$; $i<N_i$; $i++$ }
\STATE Generate random Gaussian noise $\bm{n_i}$
\IF{$\| \bm{y} - \bm{A}T(n_i)\|^2_2 < \| \bm{y} - \bm{A}T(\hat{n})\|^2_2$}
\STATE $\bm{\hat{n}} = \bm{n_i}$
\ENDIF
\ENDFOR
\FOR{$j=0$; $j<N_2$; $j++$}
\STATE $\bm{\hat{n}} = \bm{\hat{n}} + \nabla_{\bm{n}} L$
\ENDFOR
\end{algorithmic}
\end{algorithm}

\subsection{Traffic Matrix Estimation With Training Data}
In some cases, in addition to link measurements, some TMs from the past may be also available. In this case the GAN can be directly trained with the available data. In addition to the distribution of demands, the TM data may also contain spatial information that can be learned by the GAN. With the trained generator, the optimization steps will be the same as those with a distribution constraint. 

\section{Experiment Setup}\label{Sec:experiment}
We evaluate the performance of our methods with three datasets. The first dataset is the NET82 dataset which contains one TM with 82 nodes. The second dataset is the Abilene dataset \cite{zhang2003fast} which contains TMs with 12 nodes and 52 links. The third dataset is the G\'EANT dataset \cite{uhlig2006providing}, which has 23 nodes and 38 links. Note that when $\beta=1$ the Beta distribution becomes a power law distribution. In our experiments we found that the power law distribution is sufficient for fitting the distribution of the TMs. And the $\alpha$ values are the maximum likelihood estimates from the measured TMs \cite{hahn1967statistical}.

Firstly we test the performance of our method assuming only the distribution of the demands is known. For the first dataset a Beta distribution with $\alpha = 0.01154, \beta=1$ is used for the projection based method (Proj-D) and the GAN based method (GAN-D).  
 The parameters are directly used for Proj-D. For GAN-D, we first train the GAN with random matrices generated from the fitted distribution, then we use the GAN for TM estimation. 


For the Abilene dataset we use the TMs collected from March to June for distribution fitting. We use 1000 of the TMs collected in July for testing. We fit a Beta distribution with $\alpha=0.0107$, $\beta=1.0$ according to all the demands collected from March to June. Similar to the case of the first dataset, for Proj-D we use the Beta distribution directly. 

For GAN-D, the TMs from March to June are available and the TM estimation is conducted for the data in July. So the TMs from March to June can be used for the training of the GAN. The GAN is trained for 300 epochs, with 27360 TMs collected from March to June. 

For the G\'EANT dataset, a Beta distribution with $\alpha=0.01411$ and $\beta = 1.0$ is used for Proj-D. For GAN-D, 8016 TMs collected from January to March are used for the training of GAN. Network parameters for the GAN are the same as those for the Abilene dataset. Both methods are tested on 1000 TMs collected in April.

We use the same structure for the GAN for all the datasets. The generator of the GAN is a fully connected neural network with hidden layers of size 32, 64 and 128. The discriminator is also a fully connected neural network with hidden layers of size 512, 256, 256 and 256. We do not focus on finding the best parameters of the GAN in this paper. However we found it beneficial to use a larger neural network for the discriminator, so that the discriminator can more efficiently capture the difference between TMs and random matrices. ReLU is used as the activation function for the neural networks. To keep the balance between the capability of the discriminator and the generator, we update the discriminator 64 times after each training step of the generator.

\section{Performance Evaluation}\label{Sec:performance}
Performance of the methods are evaluated with two different metrics: the root mean square error (RMSE) and the normalized mean absolute error (NMAE) of the estimation results. The NMAE can be written as:
\begin{equation}
    NMAE = \frac{\| x-\hat{x} \|_1}{\| x\|_1}
\end{equation}
The results on shown in Table I. Errors are calculated for the non-zero demands. 

\begin{table}
\renewcommand{\arraystretch}{1.0}
\caption{Performance Comparison}
\label{tab:table2}
\begin{center}
    \begin{tabular}{c || c c c}
    \hline
    \multicolumn{4}{c}{Shortest Path} \\
    \hline
    Method & NET82 & Abilene & G\'EANT \\
    \hline
    Project-D (RMSE/Mbps) & 125.94 & 40.47 & 87.96 \\
    \hline
    GAN-D (RMSE/Mbps) & 194.81 & 25.74 & 65.69 \\
    \hline
    Project-D (NMAE) & 1.20 & 0.94 & 1.51 \\
    \hline
    GAN-D (NMAE) & 1.93 & 0.66 & 1.18 \\
    
    \hline
    \hline
    
    \multicolumn{4}{c}{ECMP} \\
    \hline
    Method & NET82 & Abilene & G\'EANT \\
    \hline
    Project-D (RMSE/Mbps) & 153.35 & 42.05 & 87.12 \\
    \hline
    GAN-D (RMSE/Mbps) & 191.34 & 25.74 & 65.81 \\
    \hline
    Project-D (NMAE) & 1.33 & 0.97 & 1.50 \\
    \hline
    GAN-D (NMAE) & 1.95 & 0.66 & 1.18 \\
    \hline

    \end{tabular}
    
\end{center}
\end{table}
For the NET82 dataset, Proj-D achieves RMSE of 125.94 Mbps and NMAE of 1.20. The RMSE of the results from GAN based method is 194.81 Mbps and the NMAE is 1.93. 
To evaluate the method's ability to meet the distribution constraint, we also compare the empirical cumulative distribution function (CDF) of the solutions. Figure 3 evaluates the performance of the methods on NET82. Figure 3 (a) shows the CDF of the solutions, the fitted distribution and original data. 
\begin{figure*}[!t]
    
    \centering
    \subfloat[Empirical CDF]{\includegraphics[width=.3\textwidth]{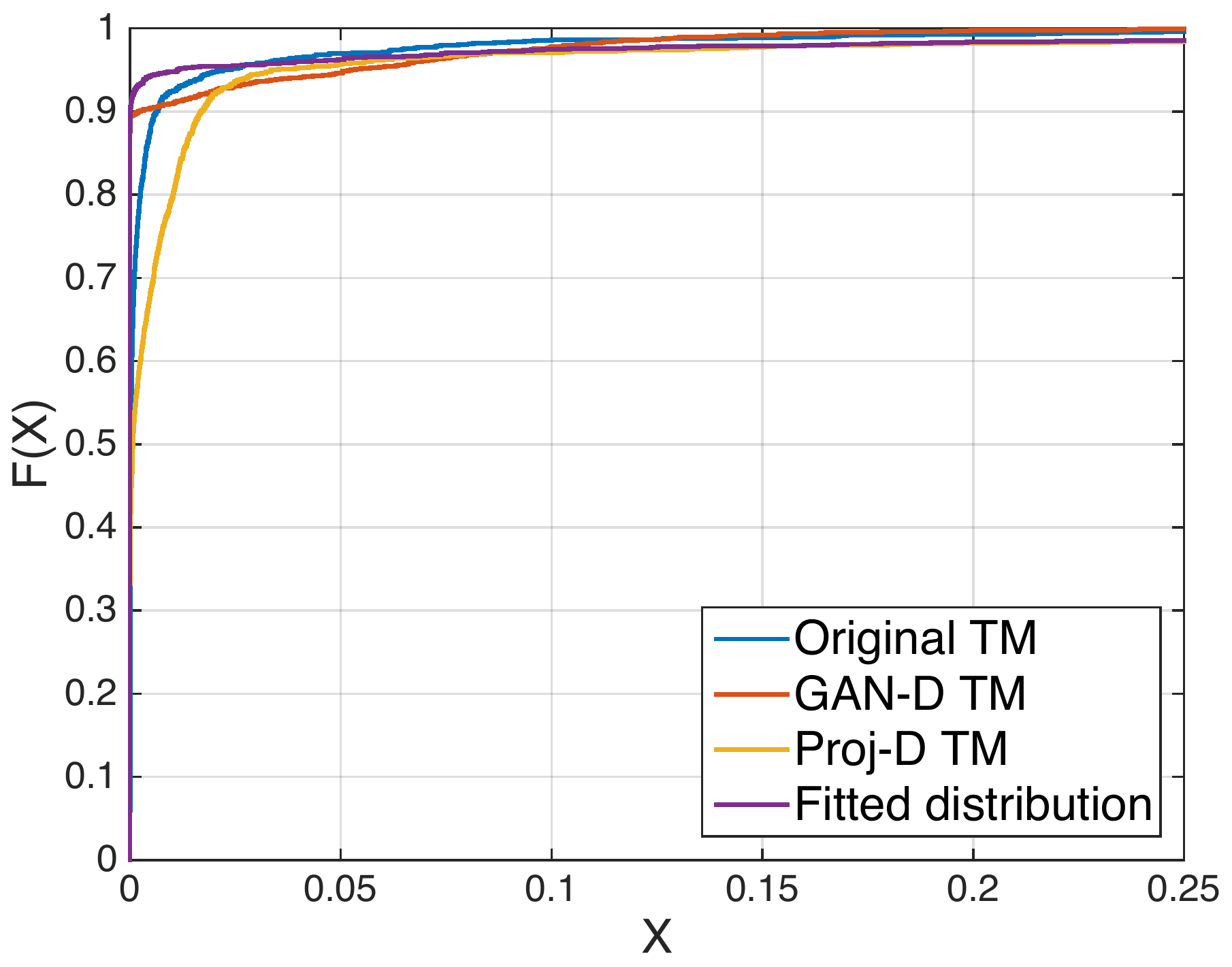}}
    \label{Fig:Traf_CDF}
    \hfil
    \subfloat[Demands]{\includegraphics[width=.3\textwidth]{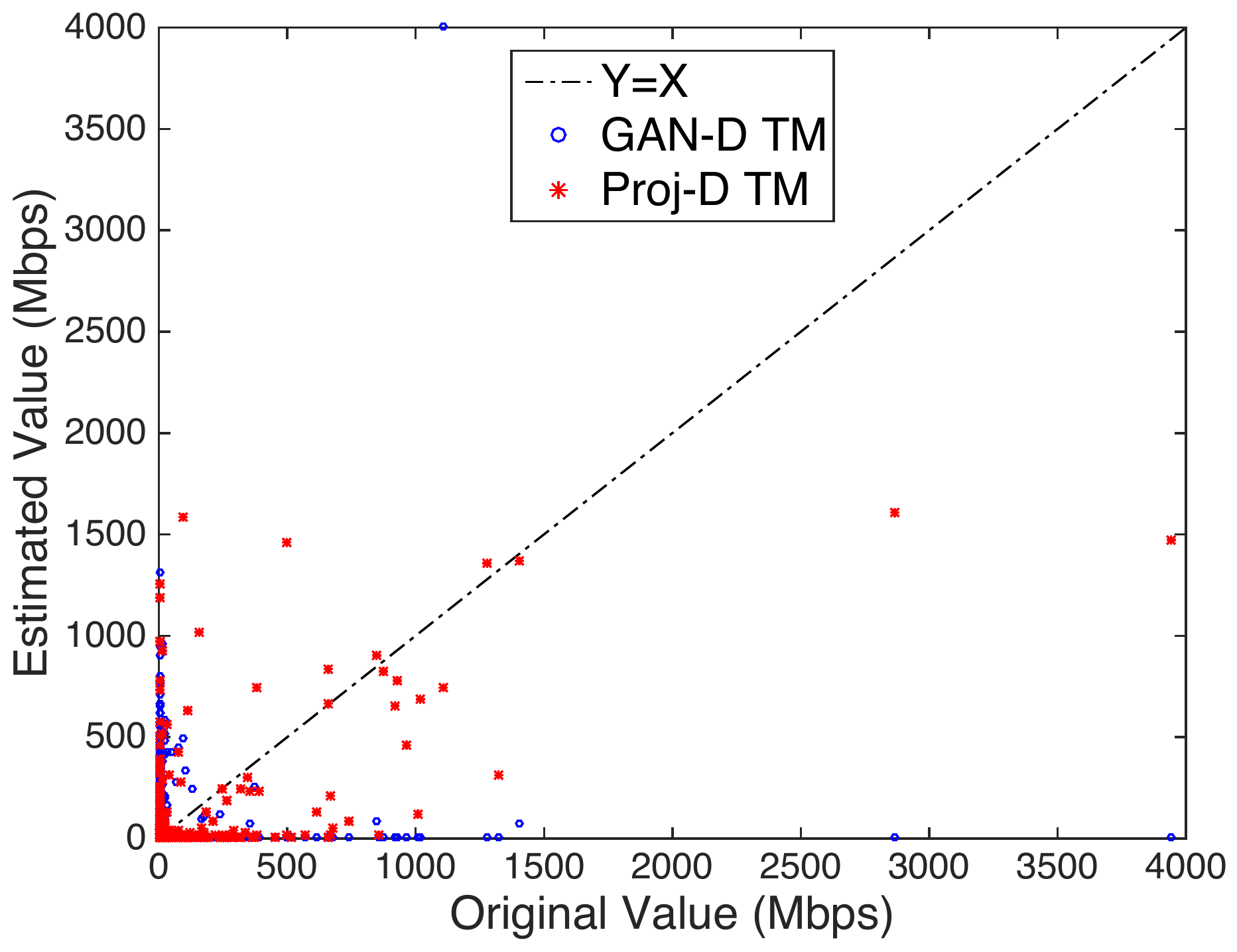}}
    \label{fig:Traf_X}
    \hfil
    \subfloat[Link measurements]{\includegraphics[width=.3\textwidth]{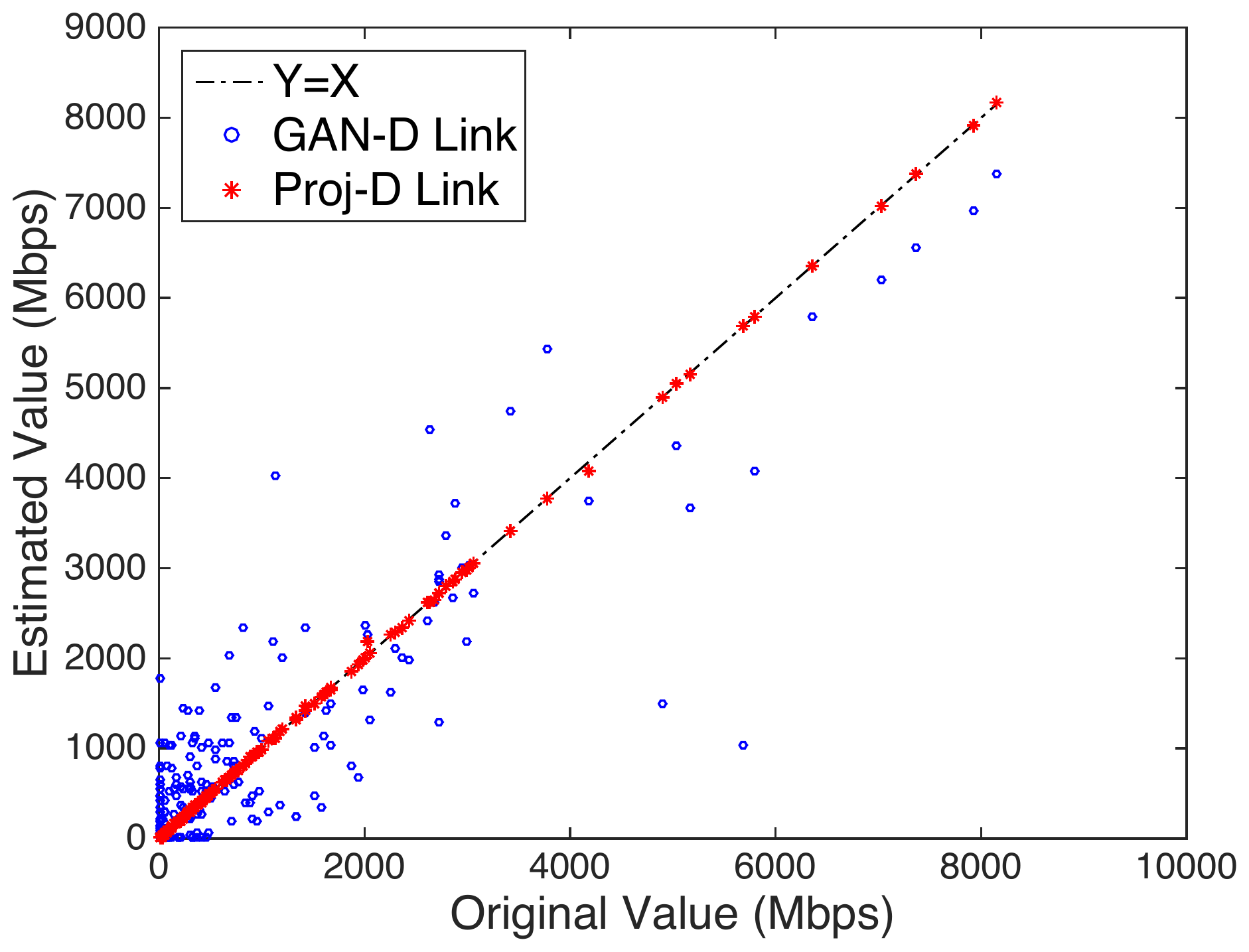}}
    \label{fig:Traf_Y}
    \caption{Performance evaluation on the NET82 dataset}
    
\end{figure*}
Figure 3 (b) shows the recovered demands versus the original demands. Figure 3 (c) shows the fitted link measurements versus the given link measurements. 

For the CDF plot, the TMs are normalized by the maximum value of all the TMs. Since very few of the normalized values are greater than 0.25, we show the CDF plot from 0 to 0.25 to better evaluate how well the estimated TMs fit the original distribution. 

Figure 4 shows the CDF, demands and link measurement of recovery results of the two methods on the Abilene dataset. The CDF plot is generated in the same way as in Figure 3. The demand plot and link measurement plot are generated from the first ten TMs. Proj-D achieves RMSE of 40.47 Mbps and NMAE of 0.94, while GAN-D achieves RMSE of 25.74 Mbps and NMAE of 0.6600. 
\begin{figure*}[!t]
    \centering
    \subfloat[Empirical CDF]{\includegraphics[width=.3\textwidth]{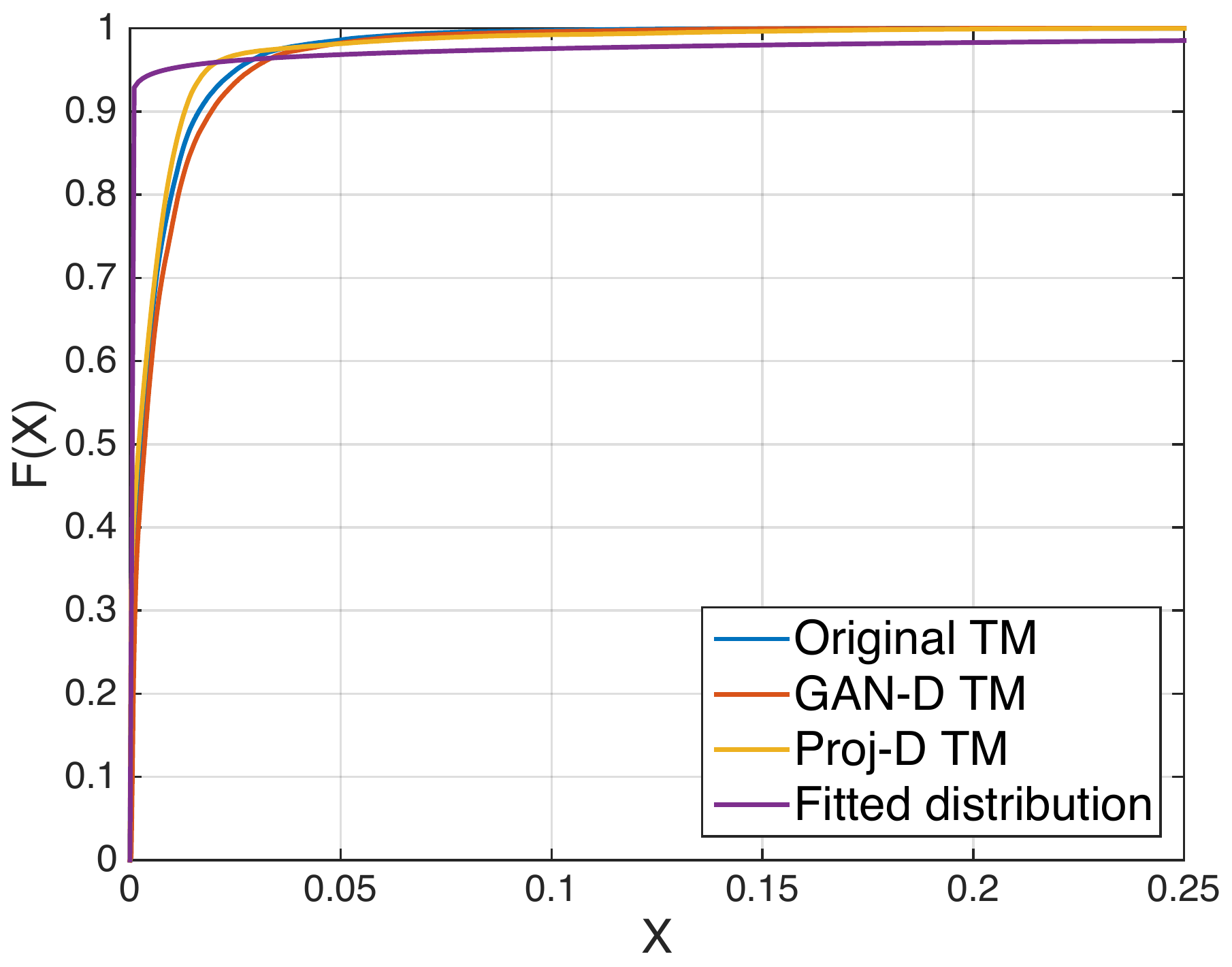}}
    \label{Fig:Ab_CDF}
    \hfil
    \subfloat[Demands]{\includegraphics[width=.3\textwidth]{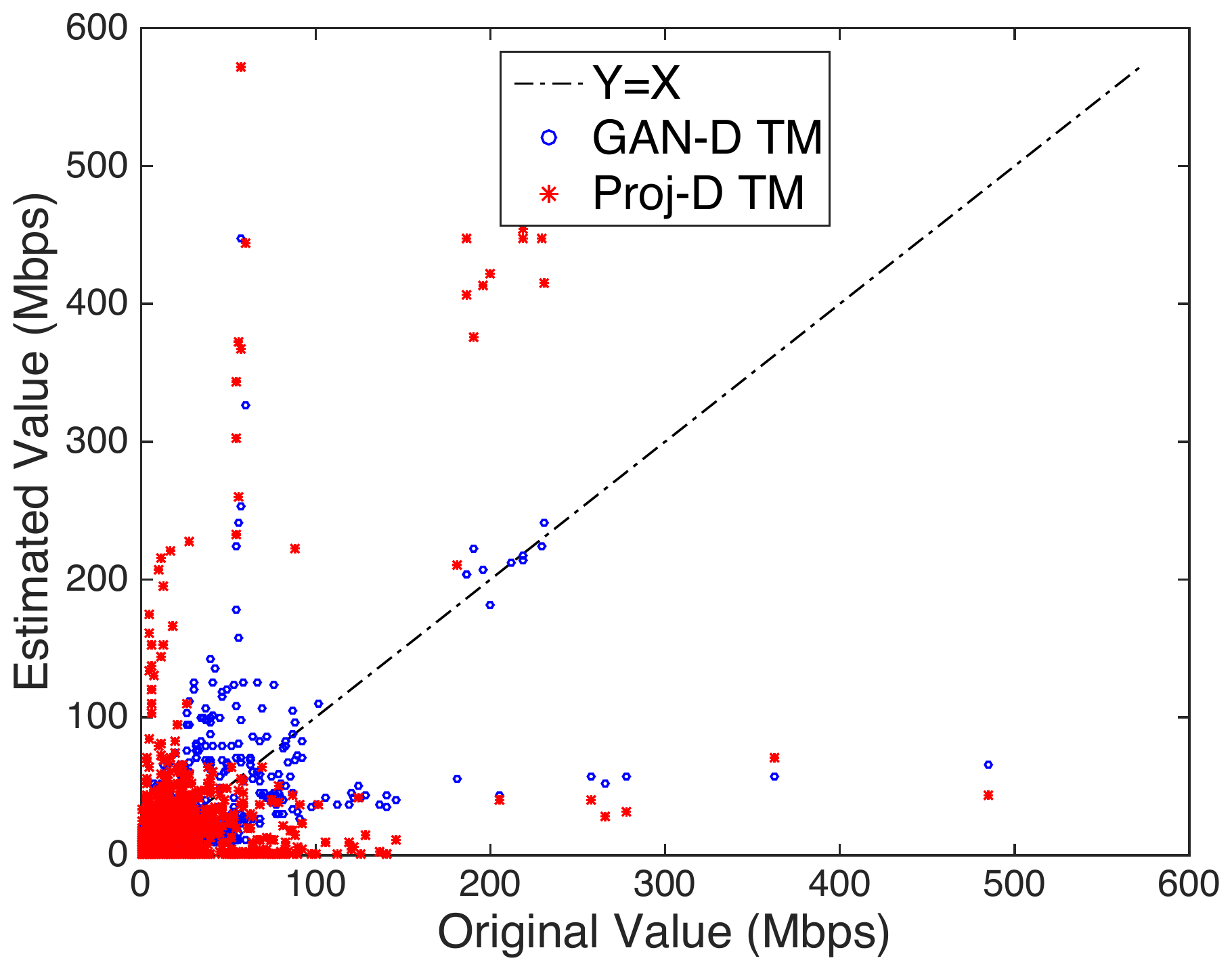}}
    \label{fig:Ab_X}
    \hfil
    \subfloat[Link measurements]{\includegraphics[width=.3\textwidth]{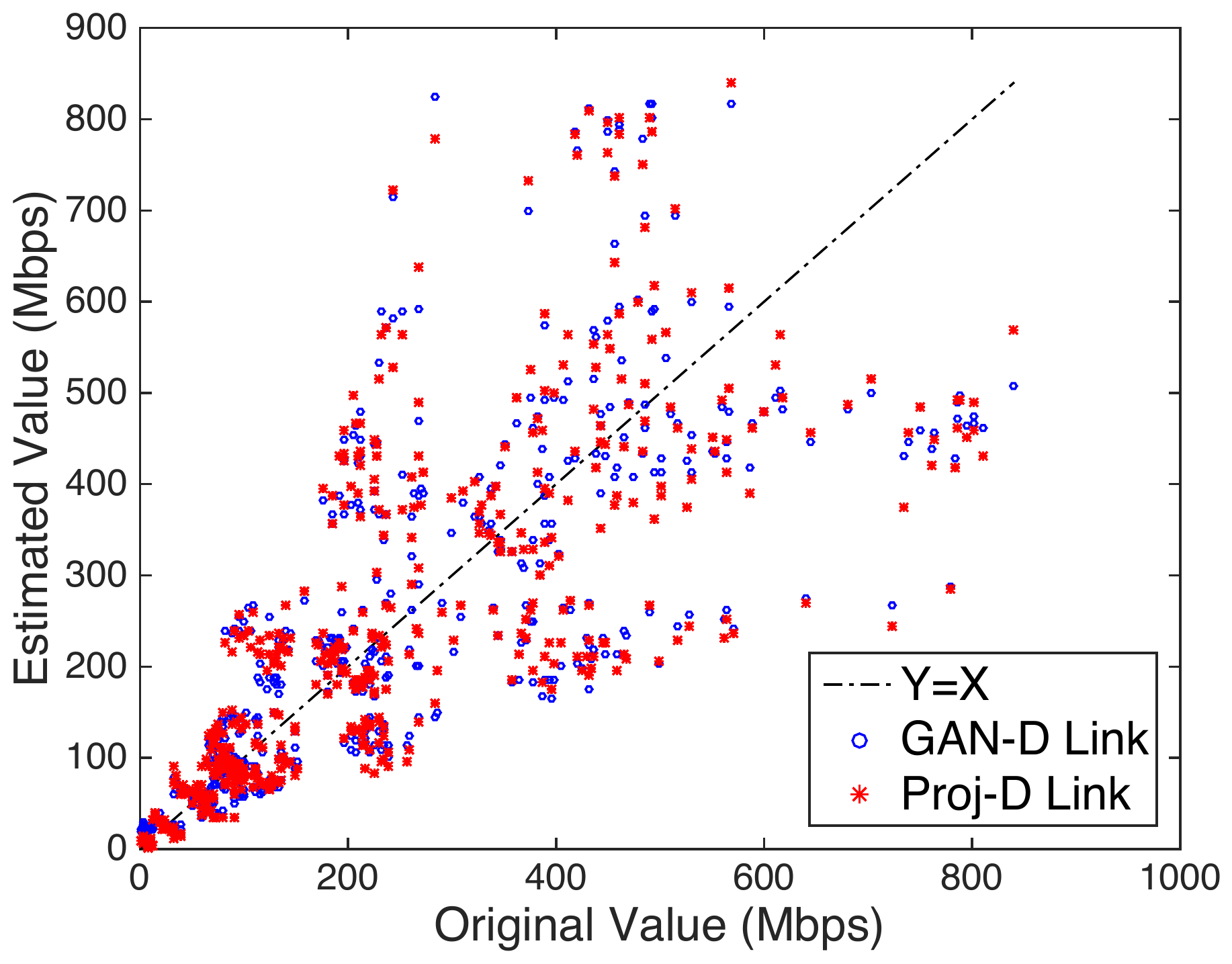}}
    \label{fig:Ab_Y}
    \caption{Performance evaluation on the Abilene dataset}
\end{figure*}

Figure 5 shows the CDF, demands and link loads of recovery results of the two methods on the G\'EANT dataset. The demand plot and link measurement plot are generated from the first ten TMs. Proj-D achieves NMAE of 1.51, GAN-D achieves NMAE of 1.18. In terms of RMSE, Proj-D has RMSE of 87.96 Mbps, GAN-D has RMSE of 65.69 Mbps. 
\begin{figure*}[!t]
    \centering
    \subfloat[Empirical CDF]{\includegraphics[width=.3\textwidth]{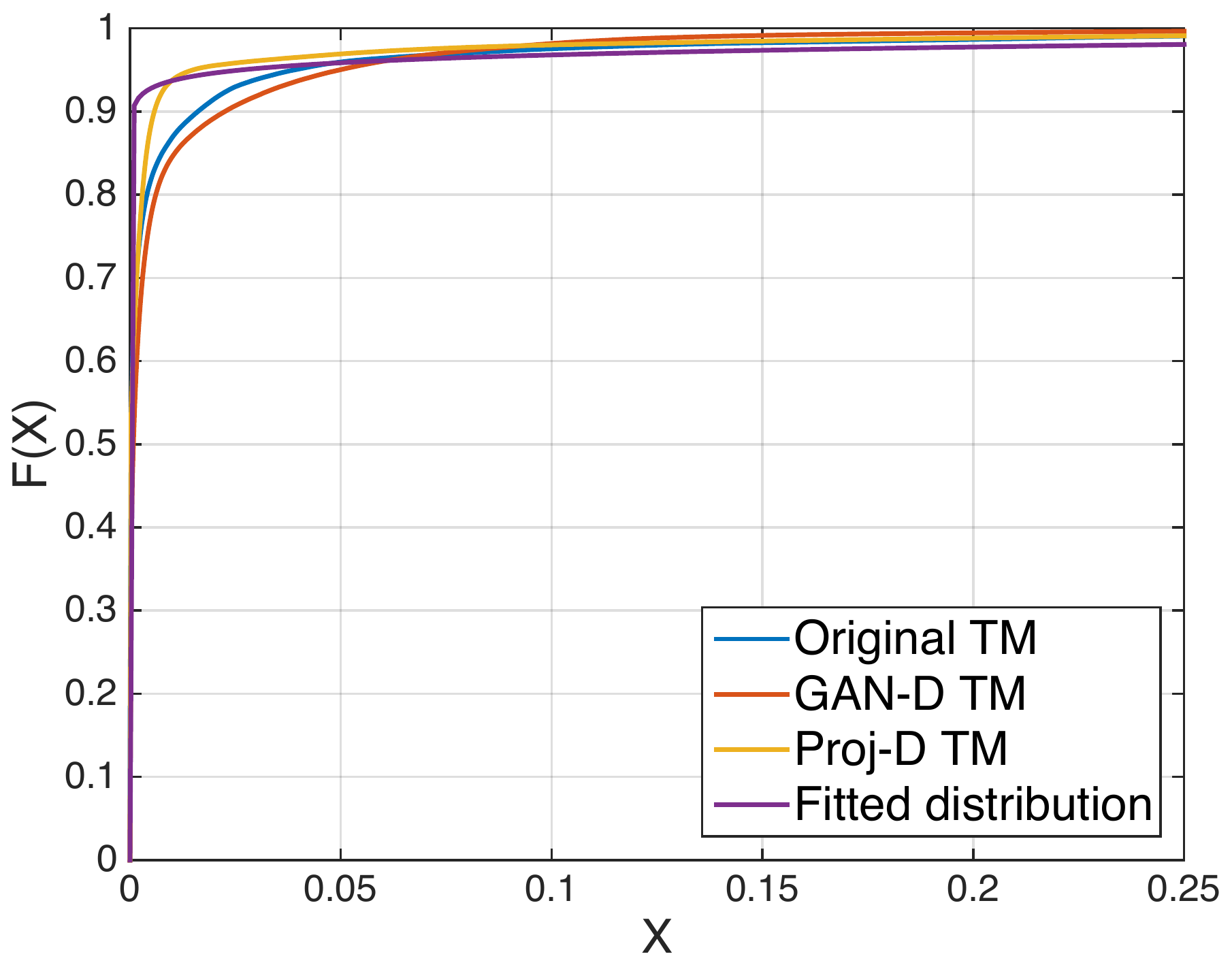}}
    \label{Fig:GT_CDF}
    \hfil
    \subfloat[Demands]{\includegraphics[width=.3\textwidth]{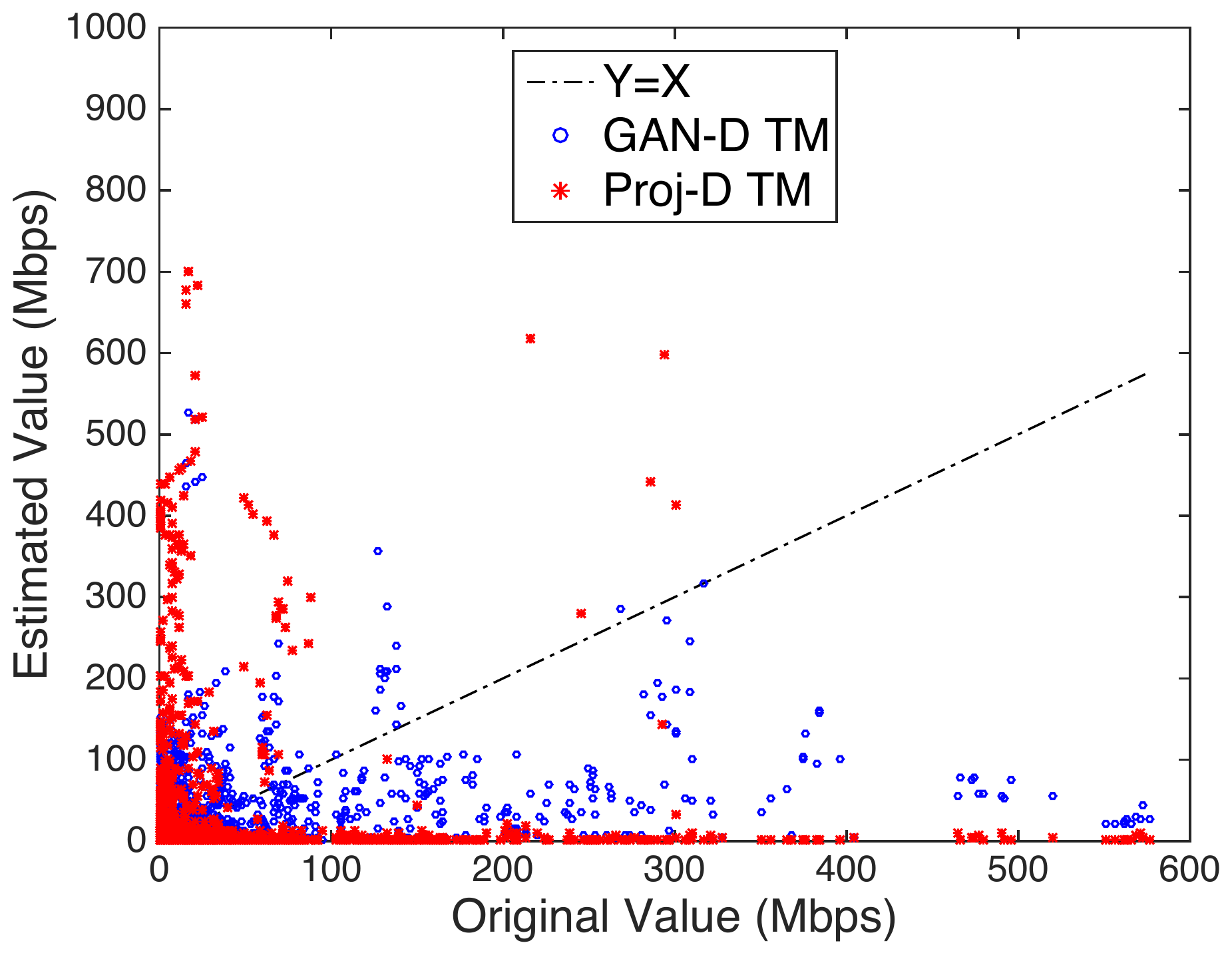}}
    \label{fig:GT_X}
    \hfil
    \subfloat[Link measurements]{\includegraphics[width=.3\textwidth]{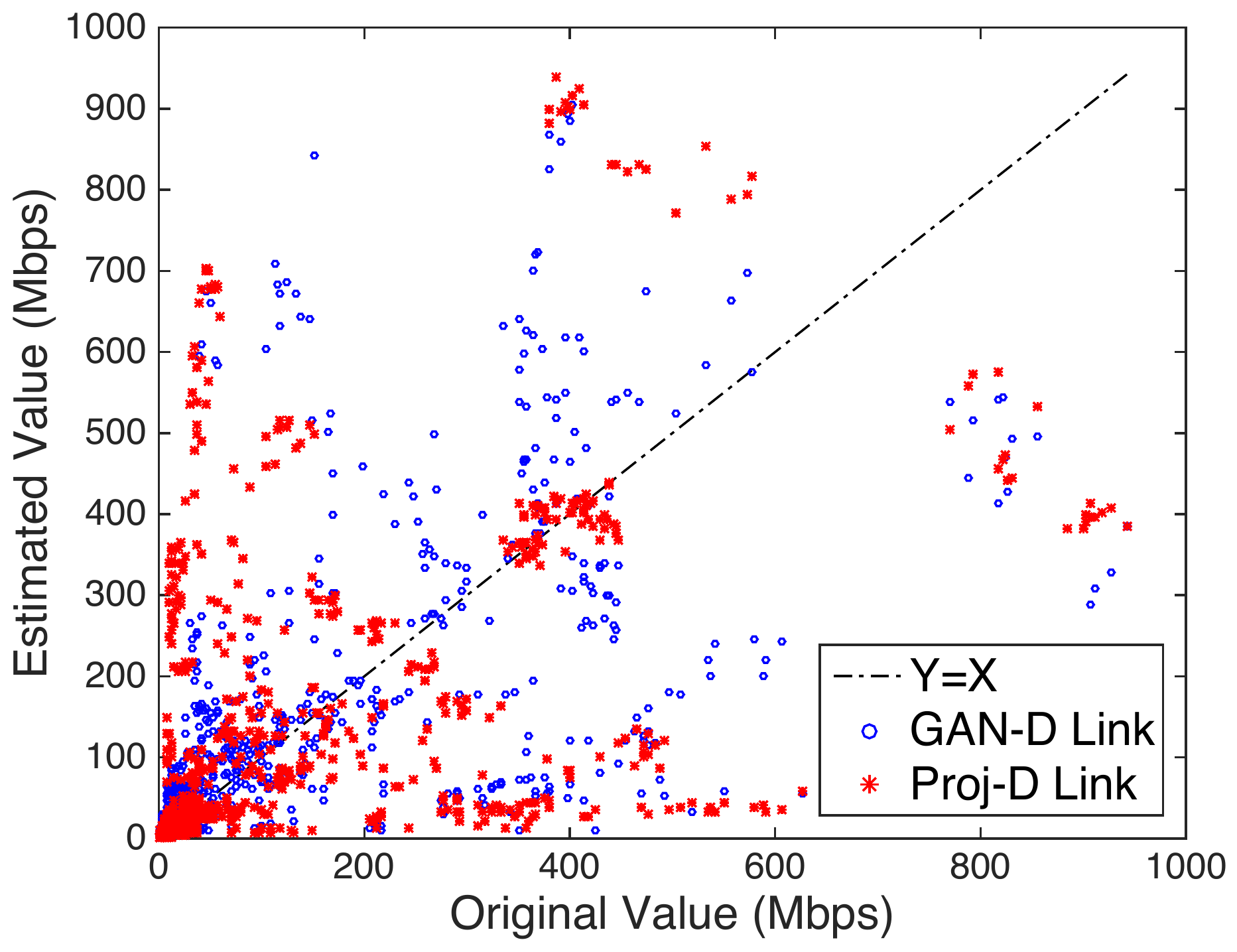}}
    \label{fig:GT_Y}
    \caption{Performance evaluation on the G\'EANT dataset}
\end{figure*}

Comparing results shown in Figure 3, for NET82, both the projection method and the GAN based method are able to provide estimation results with distributions that are similar to the fitted distribution and also the distribution of the original data. GAN-D is able to generate results that are closer to the original distribution. However Proj-D is able to generate data that fit better to the link measurement constraint, with the cost of diverting a bit from the given distribution constraint. Since the generator of the GAN is trained to generate data similar to the training set, GAN-D generates estimations that closely follow the given distribution, with the cost of worse fit of the link measurements. Though the GAN learns to generate samples according to the given distribution, it is not able to cover all possible space of the distribution, therefore GAN-D performs worse than Proj-D in terms of RMSE and NMAE.  

For the Abilene dataset, both methods are able to provide estimation results that closely meet the distribution constraint and link measurement constraints. Since the GAN is trained with real TMs measured from the past, it is able to learn the spatial correlations and other structural information of the TMs from the training data. Hence GAN-D is able to generate data that fits the distribution constraint better. For GAN-D the number of optimization steps $N_2$ also determines how well the results meet the link measurement constraints; with more optimization steps the results will fit the link measurement constraints better, but the elements of the estimated TMs will start to divert from the real value after certain number of steps. We perform the optimization process for 10000 steps, which generates results that can closely meet the link measurement constraints without too much over-fitting. For Proj-D the results can better meet the link measurement constraints, at the cost of diverting a bit from the distribution constraint. 

For the G\'EANT dataset, both Proj-D and GAN-D are able to generate results that fit the distribution constraint. This may be because that there are fewer links in this dataset so both methods can meet the distribution constraint without over-fitting. However, in terms of NMAE and RMSE the GAN based method still performs better than Proj-D. So the GAN is still able to learn spatial and structural information from the TMs used for training.

In addition to shortest path routing, Figure 6, 7 and 8 show the results with ECMP routing. Since the methods do not depend on any specific routing mechanism, they achieve similar performance with ECMP routing. 
\begin{figure*}[!t]
    
    \centering
    \subfloat[Empirical CDF]{\includegraphics[width=.3\textwidth]{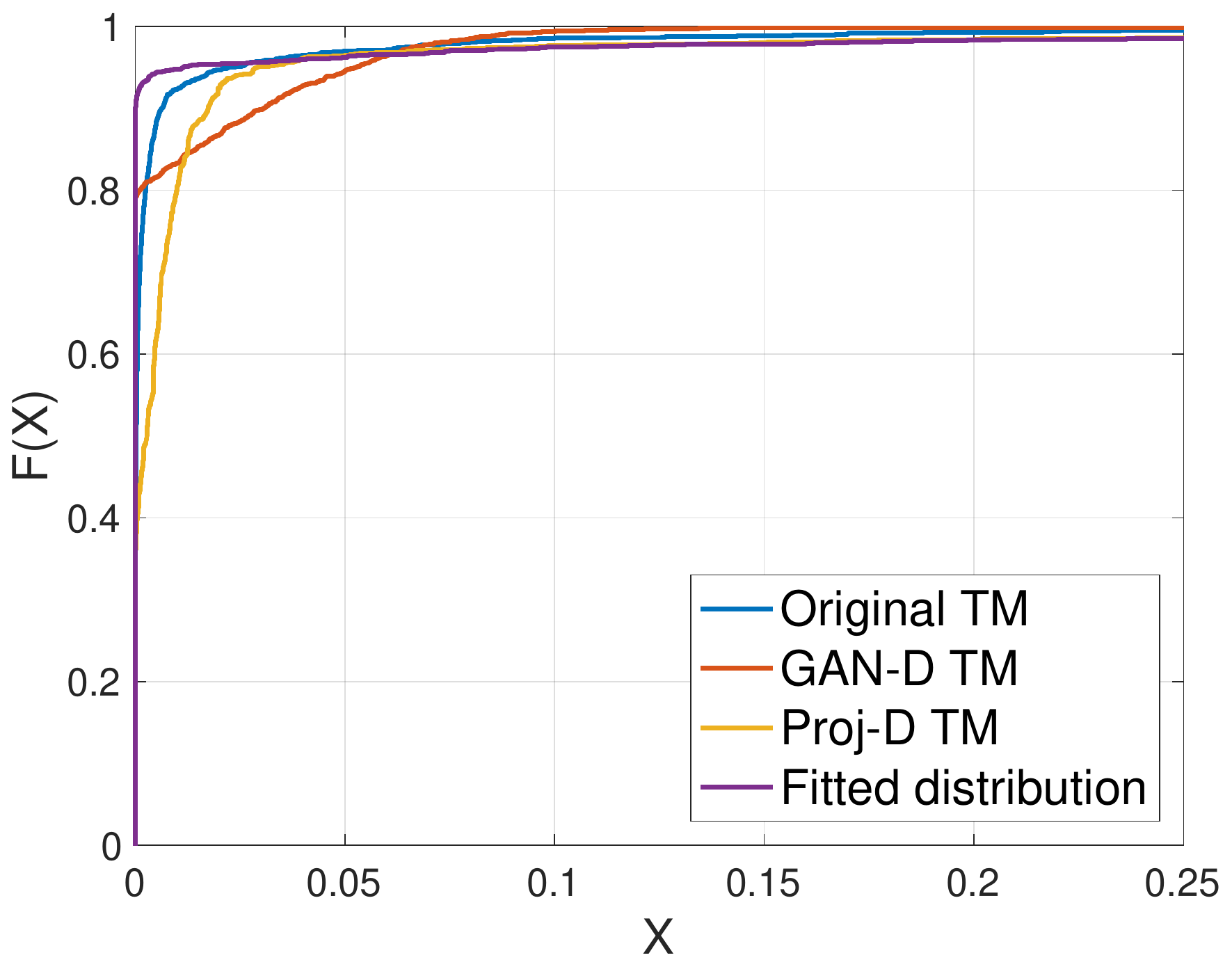}}
    \label{Fig:Traf_CDF_E}
    \hfil
    \subfloat[Demands]{\includegraphics[width=.3\textwidth]{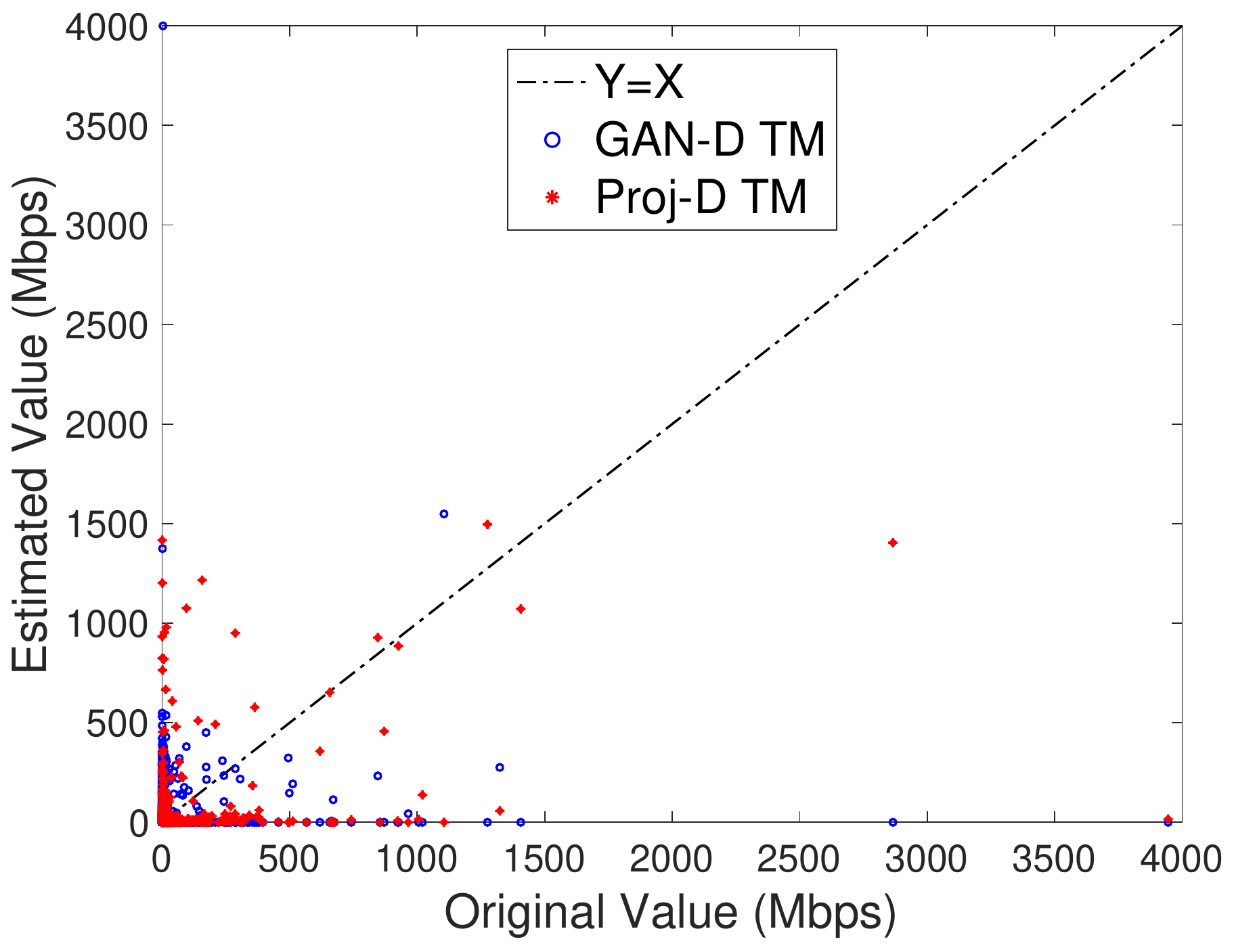}}
    \label{fig:Traf_X_E}
    \hfil
    \subfloat[Link measurements]{\includegraphics[width=.3\textwidth]{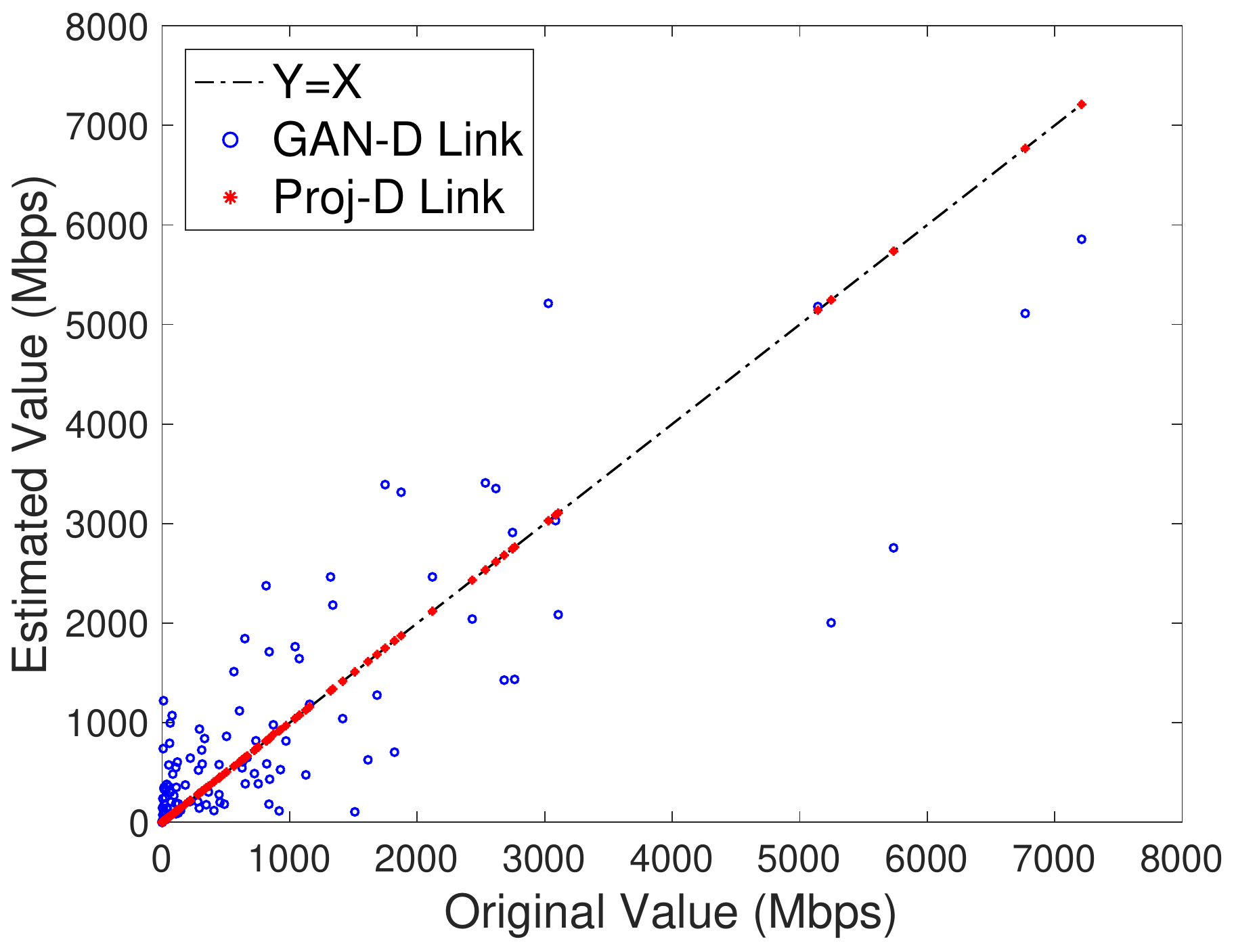}}
    \label{fig:Traf_Y_E}
    \caption{Performance evaluation on the NET82 dataset (ECMP)}
    
\end{figure*}

\begin{figure*}[!t]
    \centering
    \subfloat[Empirical CDF]{\includegraphics[width=.3\textwidth]{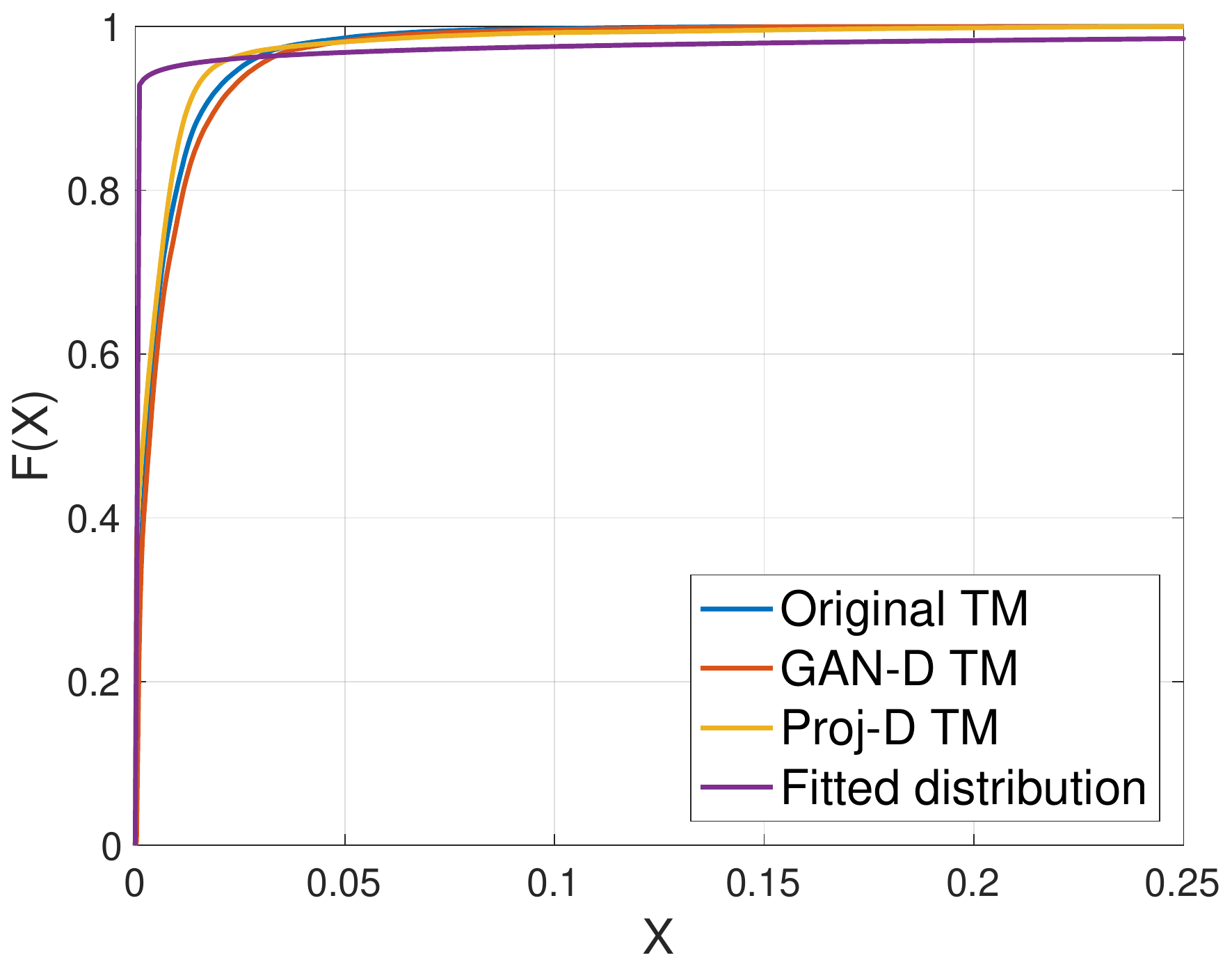}}
    \label{Fig:Ab_CDF_E}
    \hfil
    \subfloat[Demands]{\includegraphics[width=.3\textwidth]{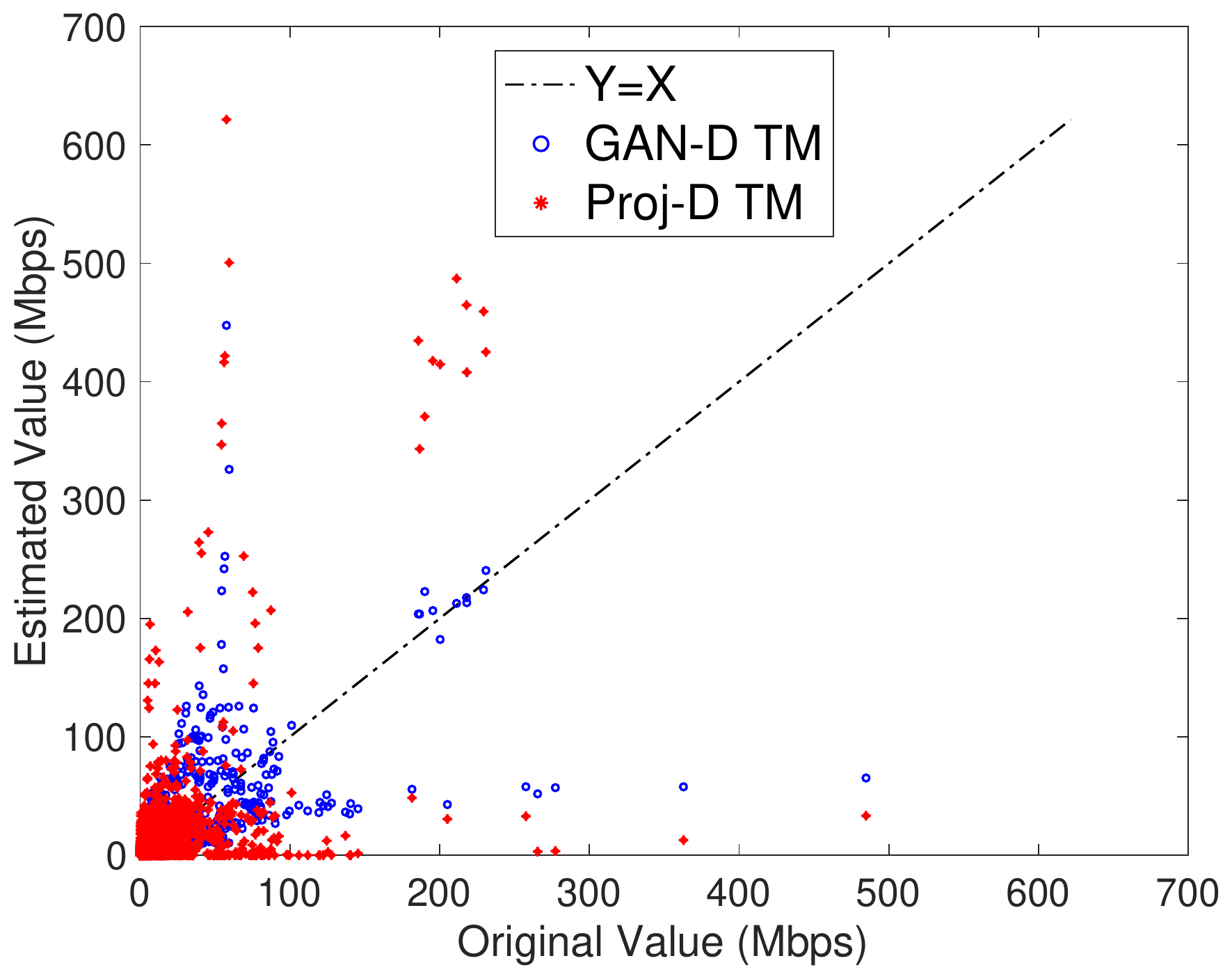}}
    \label{fig:Ab_X_E}
    \hfil
    \subfloat[Link measurements]{\includegraphics[width=.3\textwidth]{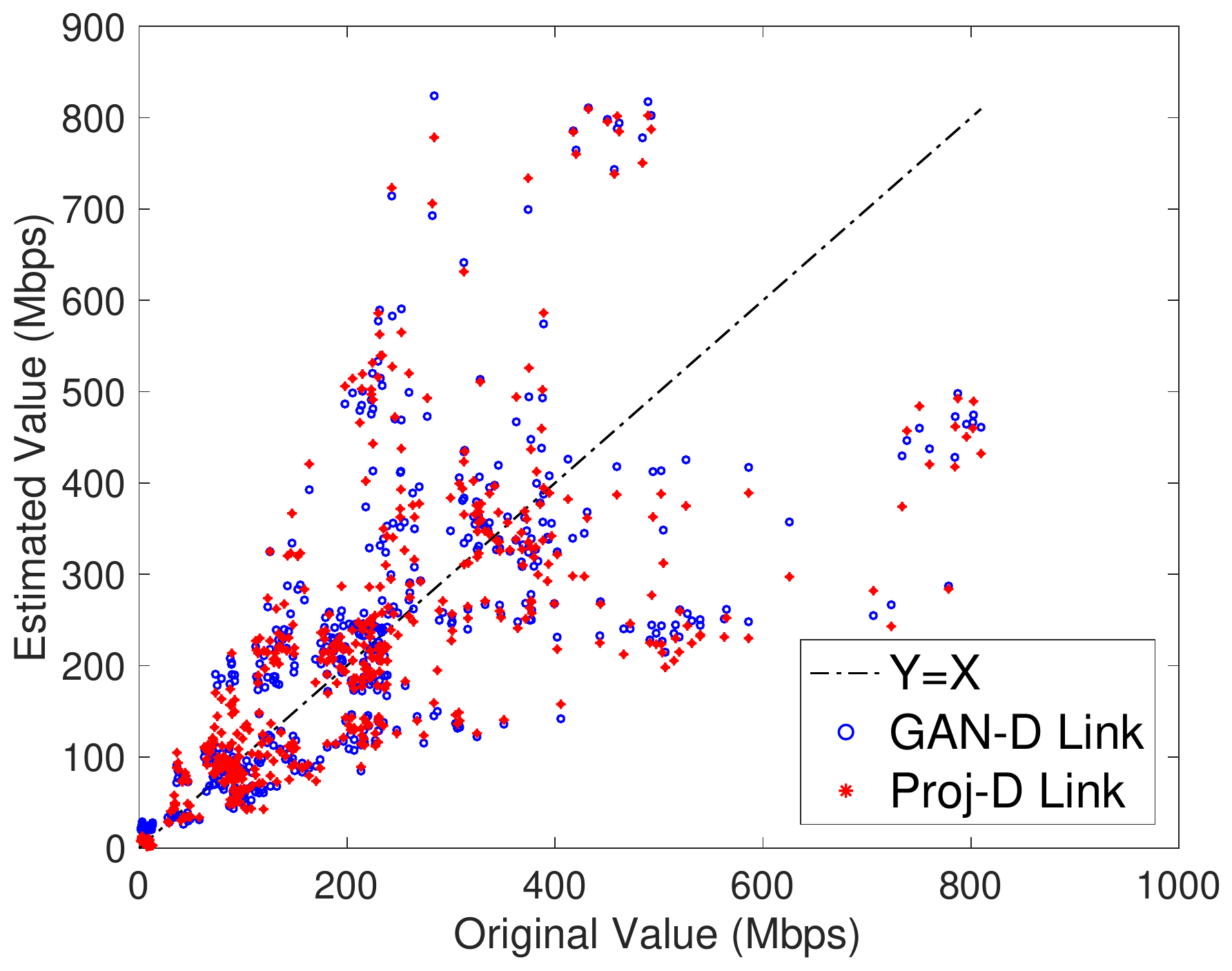}}
    \label{fig:Ab_Y_E}
    \caption{Performance evaluation on the Abilene dataset (ECMP)}
\end{figure*}

\begin{figure*}[!t]
    \centering
    \subfloat[Empirical CDF]{\includegraphics[width=.3\textwidth]{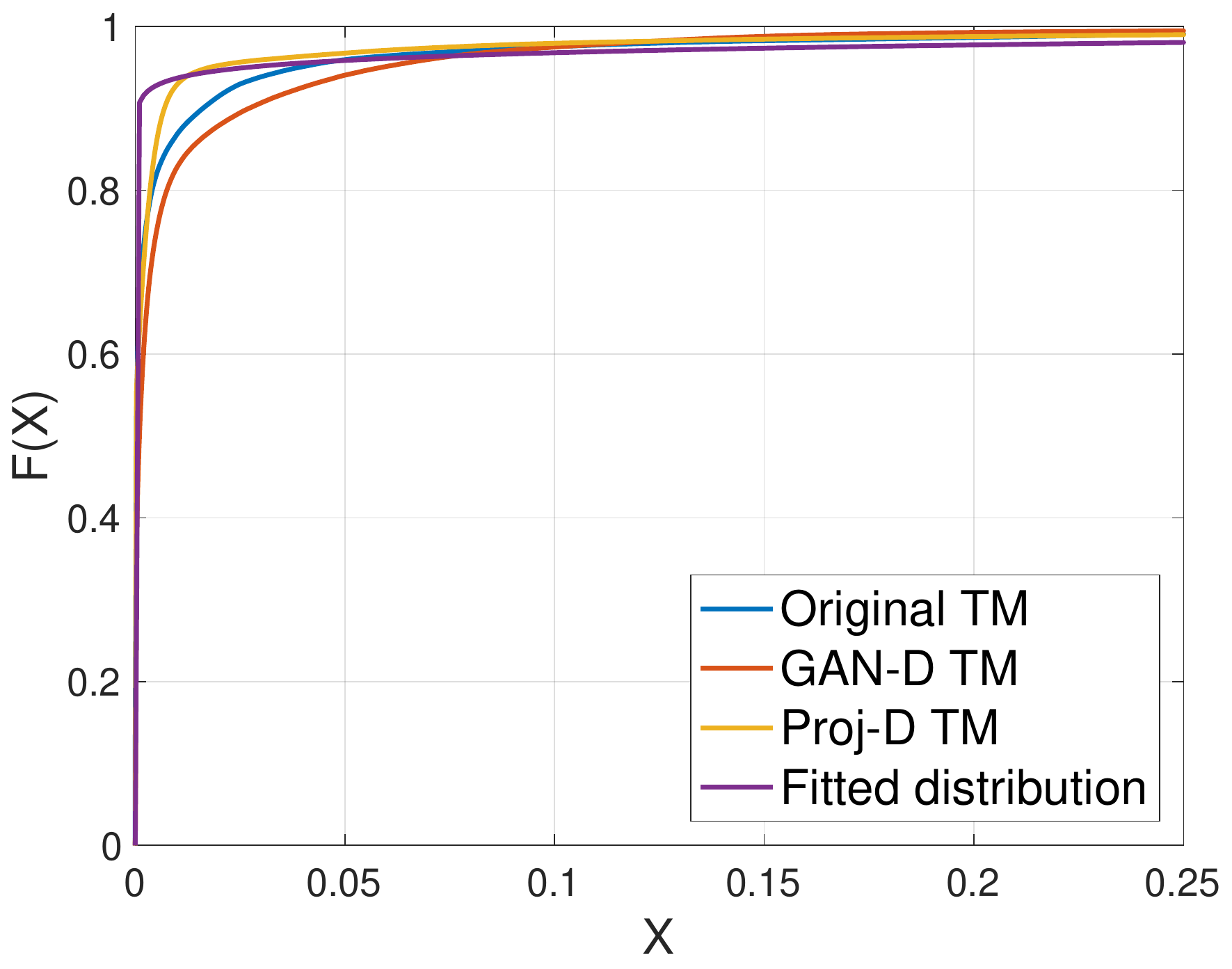}}
    \label{Fig:GT_CDF_E}
    \hfil
    \subfloat[Demands]{\includegraphics[width=.3\textwidth]{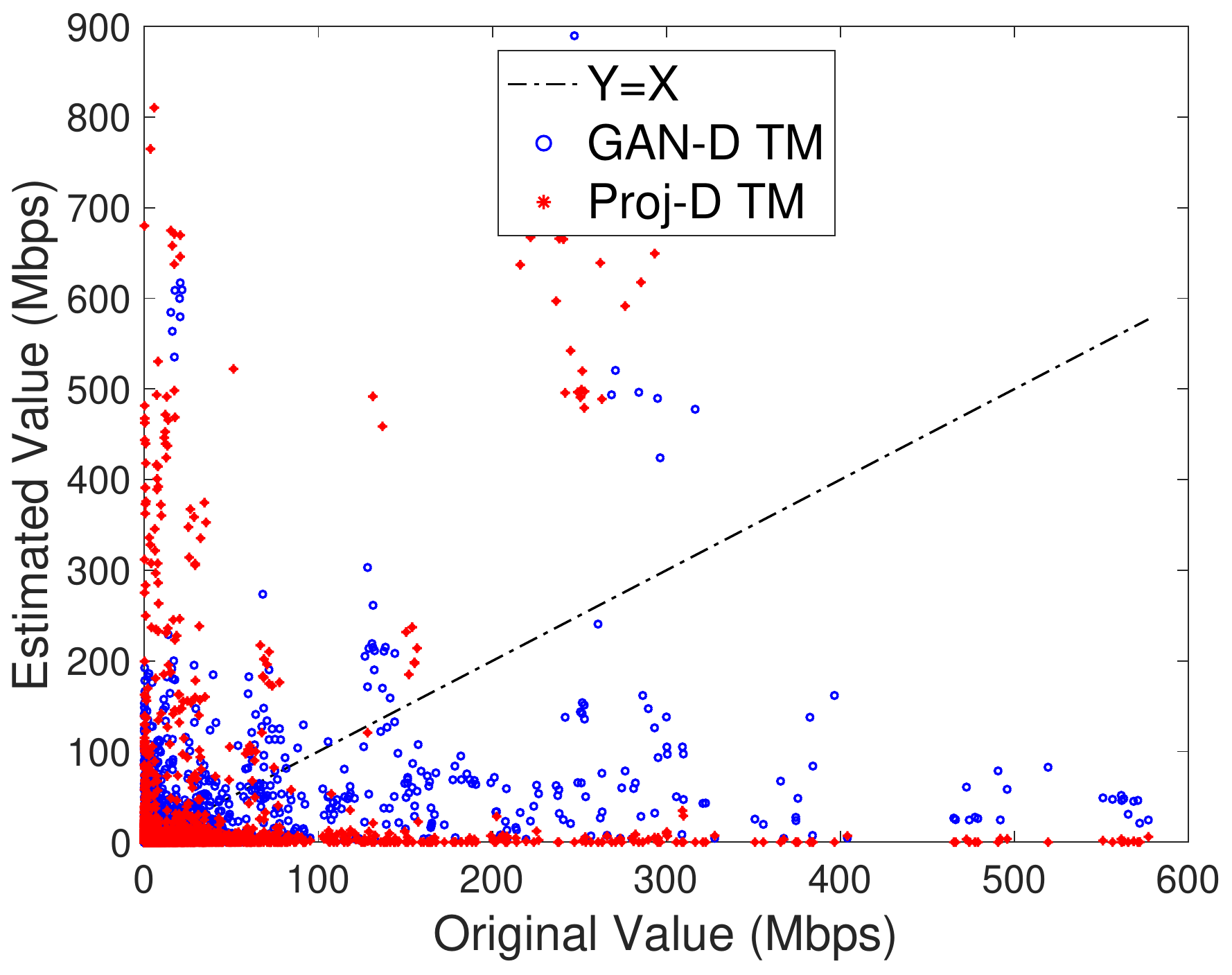}}
    \label{fig:GT_X_E}
    \hfil
    \subfloat[Link measurements]{\includegraphics[width=.3\textwidth]{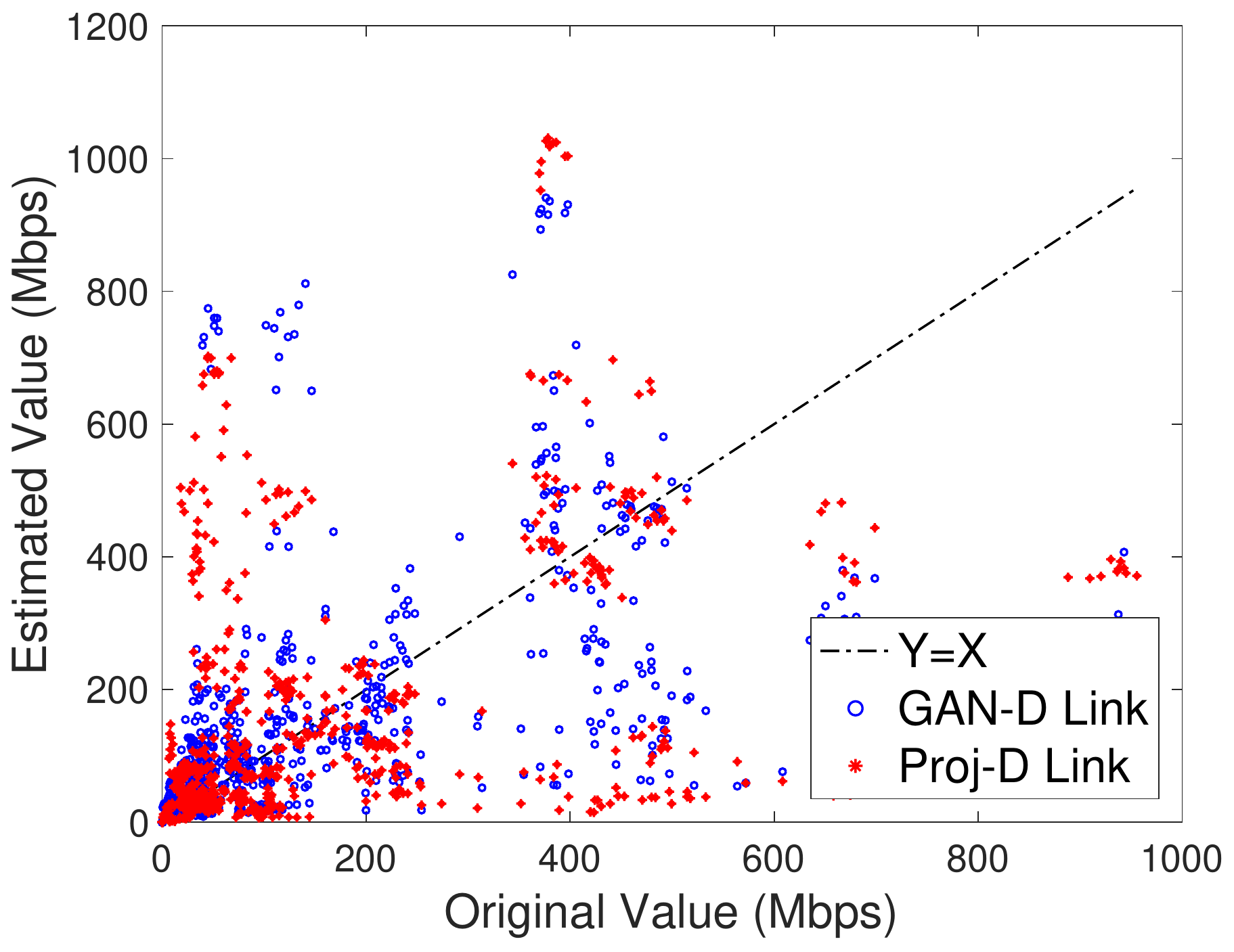}}
    \label{fig:GT_Y_E}
    \caption{Performance evaluation on the G\'EANT dataset (ECMP)}
\end{figure*}

In general, there exists the choice between meeting the distribution constraint better or meeting the link measurement constraint better. The GAN based method is able to provide estimation results that meets the distribution constraint better, without perfect fit of the link measurement constraints. The Projection based method is able to generate results that have almost exact fit of the link measurement constraints, but slightly diverge from the given distribution. The users can choose either one of the method based on their requirements in the specific use cases. 

\section{Conclusion and Future Work}\label{Sec:conclusion}
In this paper we proposed two methods for the problem of TM estimation given link measurements under a constraint of the distribution of demands. Experiment results show that both the method Proj-D and GAN-D are able to generate estimation results that fit the link measurements and the distribution constraint. The Projection based method is able to provide estimation results that fits the link constraints better, while the GAN based method generates TMs that better fit the given distribution. In addition, if TMs measured in the past are available, the GAN based method is able to learn the spatial and structural correlations of the TM data and provide better estimation results. \balance Future work includes extending these methods to other similar problems and finding the suitable kind of GAN for the GAN based method.

\bibliographystyle{IEEEtran}
\bibliography{Citation}
\end{document}